\documentclass{LMCS}

\def\doi{7 (4:05) 2011}
\lmcsheading%
{\doi}
{1--25}
{}
{}
{Mar.~31, 2010}
{Nov.~23, 2011}
{}

\usepackage{enumerate}
\usepackage{hyperref}

\usepackage{proof}    
\usepackage{verbatim}

\usepackage{url}
\usepackage{amsmath}[1999/12/15] 
\usepackage{amssymb}
\usepackage{amsthm}              
\usepackage{stmaryrd}            
\usepackage{relsize}             
\let\citep=\cite
\theoremstyle{plain}

\theoremstyle{definition}
 \newtheorem{bsp}[thm]{Example}
 \newtheorem{remm}[thm]{Remark}

\def\calc{\mathcal{C}}
\def\cald{\mathcal{D}}
\def\calf{\mathcal{F}}
\def\cali{\mathcal{I}}
\def\calif{\mathcal{IF}}

\def\calp{\mathcal{P}}

\def\calv{\mathcal{V}}
\def\calw{\mathcal{W}}

\def\consc{\textbf{c}}
\def\consd{\textbf{d}}

\newcommand*{\setn}{\mathbb{N}}

\def\oto{\overrightarrow}

\def\lrbk#1{\llbracket#1\rrbracket}
\def\cons#1{\langle #1\rangle}

\def\al{\alpha}
\def\be{\beta}
\def\de{\delta}
\def\De{\Delta}
\def\ga{\gamma}
\def\Ga{\Gamma}
\def\ka{\kappa}
\def\la{\lambda}

\newcommand{\edash}{\vdash}
\def\ed{\edash}
\def\vd{\vdash}

\newcommand{\dom}{{\mathrm d}{\mathrm o}{\mathrm m}}

\newcommand{\tm}{\subseteq}
\def\tminus{\text{-}}

\def\app{\textup{\textsf{app}}}
\def\lam{\textup{\textsf{lam}}}

\def\menge#1#2{\{ #1 \, | \, #2\}}
\def\single#1{\{#1\}}

\newcommand{\rbm}[1]{\raisebox{-1.5ex}[0.5ex]{$#1$}}
\newcommand{\rbmm}[1]{\raisebox{-2.7ex}[0.7ex]{$#1$}}

\newcommand*{\equote}[1]{\textquotedblleft#1\textquotedblright}

\def\cc{\textup{\textsc{CC}}}

\def\pifn#1#2{\Pi #1\,.\,#2}

\def\bs{\backslash}

\def\arity{\textup{\textsf{Arity}}}

\def\case#1{\textup{\textsf{case}}(#1)}
\def\fix#1#2{\textup{\textsf{fix}}\,\,#1\,\{#2\}}

\def\ind{\textup{\textsf{Ind}}}
\def\indic#1{\ind_{#1}\hspace{-.05cm}\single{\De_I:=\De_C}}

\def\lett#1#2{\textup{\textsf{let}}\,\,#1\,\, \textup{\textsf{in}}\,\,#2}

\def\wf{\calw\calf}

\def\tnat{\textup{\texttt{nat}}}
\def\tprop{\textup{\texttt{Prop}}}
\def\tset{\textup{\texttt{Set}}}
\def\ttype{\textup{\texttt{Type}}}
\def\ttt#1{\textup{\texttt{#1}}}

\def\uscore{\mathunderscore}

\def\lh{\mathit{lh}}

\def\sma#1{<\!\!\!#1}
\def\equ#1{=\!\!\!#1}

\begin{document}

\title[Proof-irrelevant model of CC with pred. induction and judg. equality]{Proof-irrelevant model of  CC with predicative induction and judgmental equality}

\author[G.~Lee]{Gyesik Lee\rsuper a}	
\address{{\lsuper a}Hankyong National University, Anseong-si, Kyonggi-do, Korea}	
\email{gslee@hknu.ac.kr}  
\thanks{{\lsuper a} Corresponding author:\
For Gyesik Lee, this work was partly supported by Mid-career Researcher Program through NRF funded by the MEST (2010-0022061).}

\author[B.~Werner]{Benjamin Werner\rsuper b}	
\address{{\lsuper b} INRIA Saclay and LIX, Ecole Polytechnique, 91128 Palaiseau Cedex, France}	
\email{Benjamin.Werner@inria.fr}  






\begin{abstract}
  \noindent We present a set-theoretic, proof-irrelevant model for Calculus of Constructions (CC) with predicative induction and judgmental equality in Zermelo-Fraenkel set theory with an axiom for countably many inaccessible cardinals. 
  We use Aczel's trace encoding which is universally defined for any function type, regardless of being impredicative. Direct and concrete interpretations of simultaneous induction and mutually recursive functions are also provided by extending Dybjer's interpretations on the basis of Aczel's rule sets.
  Our model can be regarded as a higher-order generalization of the truth-table methods. We provide a relatively simple consistency proof of type theory, which can be used as the basis for a theorem prover.
\end{abstract}

\maketitle


\section{Introduction}\label{S:one}

\subsection*{Informal motivation}
The {\it types-as-sets} interpretation of type theory in a sufficiently strong classical axiomatic set theory, such as the Zermelo-Fraenkel (ZF) set theory, has been regarded as the most straightforward approach to demonstrating the consistency of type theory (cf. \cite{aczel-relating} and \cite{coquand-meta}). It can be construed as a higher-order generalization of the truth-table methods. Such a model captures the intuitive meaning of the constructs: the product, $\la$-abstraction, and application correspond to the ordinary set-theoretic product, function, and application, respectively.

A straightforward model of type theory is very useful for establishing the consistency of type theory, and it can be used to determine the proof-theoretic strength of type theory (cf. \cite{aczel-relating,dybjer,dybjer00,werner-set}). However, a higher-order generalization of the trivial Boolean model is {\it not so simple} (cf. \cite{miwe}). The main cause of this problem, as identified by \cite{reynolds}, is the fact that type systems containing Girard-Reynolds' second-order calculus cannot have the usual set-theoretic interpretation of types. The only way to provide a set-theoretic meaning for an impredicative proposition type is to identify all the proof terms of that proposition type: Proposition types are interpreted either by the empty set or a singleton with a canonical element. Thus, proof-irrelevant models are necessary for interpreting reasonable higher-order type systems.

Set-theoretic models of type theory can be understood in a straightforward manner. \cite{werner-proof} showed that they can be used as the basis of proof assistants in programming with dependent types. This is because they provide a mechanism to distinguish between computational and logical parts. 
Werner's system is a proof-irrelevant version of Luo's Extended Calculus of Constructions (ECC; \citep{luo89-lics}), and the set-theoretic model is an extension of that of Calculus of Constructions (CC) defined by \cite{miwe}.

Luo's ECC is a Martin-L\"of-style extension of CC, with strong sum types and a fully cumulative type hierarchy. At the lowest level, there is an impredicative type $\tprop$ of propositions. This is followed by a hierarchy of predicative type universes $\ttype_i$, $i = 0, 1, 2,\dots$:

\begin{iteMize}{$\bullet$}
\item $\tprop$ is of type $\ttype_0$;

\item $\ttype_i$ is of type $\ttype_{i+1}$;

\item $\tprop \prec \ttype_0 \prec \ttype_{1} \prec \cdots$.
\end{iteMize}

Werner's system, however, does not include the subtyping rule $\tprop \prec \ttype_0$, which could complicate the model construction, as identified by \cite{miwe}. Their model constructions cannot be extended to ECC. We will explain this in detail in Remark \ref{counter-subtyping}.

In this paper, we investigate the inclusion of $\tprop \prec \ttype_0$, and we show that type theory with judgmental equality, {\it \`{a} la} \cite{ml-84}, can have a simple proof-irrelevant model. We expect our results to play a key role in the theoretical justification of proof systems based on Martin-L\"of-style type theory.

\subsection*{Overview of the work}
Martin-L\"of type theory and Logical Framework include typing rules for the equality of objects and types:\medskip
\[
\Ga \vd M = N : A
\qquad \text{and} \qquad
\rbm{\infer[(conv)]{\Ga \vd M : B}{\Ga \vd M : A & \Ga \vd A =  B}}
\]
In particular, Barendregt's PTS-style $\be$-conversion side condition
turns into an explicit judgment. Two objects are not just equal; they are equal with respect to a type (cf. \cite{nordstroem,goguen-phd,aczel-relating}).

The type system considered in our study is CC with predicative induction and judgmental equality. It is a type system with the following features: dependent types, impredicative type $(\tprop)$ of propositions, a cumulative hierarchy of predicative universes $(\ttype_i)$, predicative inductions, and judgmental equality.

The main difficulty in the construction of a set-theoretic model of our system stems from the impredicativity of $\tprop$ and the subtyping property $\tprop \prec \ttype_0$. Without subtyping, one could use the solution provided by \cite{miwe} and \cite{werner-proof}, whereby proof-terms are syntactically distinguished from other function terms. Thus, the problem lies in the case distinction between the impredicative type $\tprop$ and the predicative types $\ttype_i$, whereas the subsumption eliminates the difference. An interpretation function $f: \single{0, 1} \to \calv$ is required, where $\calv$ is a set universe, that is different from the identity function. See Section \ref{interpretation} for further details.

For a set-theoretic interpretation of the cumulative type universes and predicative inductions, it is sufficient to assume countably many (strongly) inaccessible cardinals. \cite{werner-set} showed that ZF with an axiom guaranteeing the existence of infinitely many inaccessible cardinals is a good candidate. However, it is not clear whether the inaccessible cardinal axiom is necessary for our construction. The required feature of an inaccessible cardinal $\ka$ is the closure property of the universe $\calv_\ka$ under the powerset operation. This is a necessary condition for the interpretation of inductive types. Following \cite{dybjer}, we use Aczel's rule sets to obtain a direct and concrete interpretation of induction and recursion rules.

The remainder of this paper is organized as follows. In Section \ref{cc}, we provide a formal presentation of CC with predicative induction and judgmental equality. Examples are presented to enable the reader to understand the syntax and typing rules. This section can be regarded as an introduction to the base theory of the proof assistant Coq. Indeed, the syntax we have provided is as close to Coq syntax as that used in practice, except for the judgmental equality and the restriction on predicative inductions.\footnote{We remark that many impredicative inductive types can be coded by impredicative definitions (cf. \cite{girard-proofs,coquand-meta,werner-set}).}

The difficulties in providing set-theoretic interpretations of impredicative or polymorphic types, subtypes, etc., are discussed in Section \ref{interpretation}. We use the computational information about the domains saved in the interpretation of $a : A$ to avoid these difficulties. This means that for the construction of set-theoretic models, type systems with judgmental equality are more explicit than systems without it. Using some typical examples, we explain the construction of a set-theoretic interpretation of inductive types and recursive functions. 

Finally, in Section \ref{soundness}, we prove the soundness of our interpretation. The proof itself is relatively simple, and it can also be used to verify the consistency of our system. This is because some types such as $\Pi (\al:\tprop).\al$ will be interpreted as the empty set; hence, they cannot be inhabited in the type system.

In Section \ref{conclusion}, we summarize the main results, and we discuss related work for future investigation.

\section{Formal presentation of CC with judgmental equality}\label{cc}

First, we provide the full presentation of the system, i.e.,  Coquand's $\cc$ with judgmental equality and predicative induction over infinitely many cumulative universes.

\subsection{Syntax}
We assume an infinite set of countably many variables, and we let $x, x_i, X, X_i,...$ vary over the variables. We also use special constants $\tprop$ and $\ttype_i$, $i\in\setn$. They are called \emph{sorts}. Sorts are usually denoted by $s, s_i, $ etc.\footnote{In this paper, we do not consider the sort $\tset$. Indeed, when (the impredicative or predicative sort) $\tset$ is placed at the lowest level in the hierarchy of sorts, as in the case of the current development of Coq, there is no way to provide a universal set-theoretic interpretation of both $\tset$ and $\tprop$, as identified by \cite{reynolds}. Note, however, that $\ttype_0$ in our system plays the role of the predicative $\tset$.}

\begin{defi}[Terms and contexts]
  The syntax of the objects is given as follows.
\[
\begin{array}{rcll}
t, t', t_i, A, A_i & ::= & x \mid s \mid \Pi x:t. t' \mid \la x:t. t' \mid \lett{x := t}{t'} \mid t\, t' & (\textit{terms})\\[.5ex]
  & & \mid \case{t,t',\vec{t}\,} \mid \ind_n\single{\De:=\De'}\cdot x  & \\[.5ex]
  & & \mid \fix{x_i}{x_0 / k_0: A_0:=t_0,\dots, x_n / k_n: A_n:= t_n} & \\[2ex]

\De, \De' & ::= & [\,] \mid \De, (x:t) & (\textit{declarations})\\[2ex]

\Ga, \Ga' & ::= & [\,] \mid \Ga, (x:t) \mid \Ga, (x:=t:t') \mid \Ga, \ind_n\single{\De:= \De'} & (\textit{contexts})
\end{array}
\]
Here, $[\,]$ denotes the empty sequence.
\end{defi}

\begin{defi}[Atomic terms]
\emph{Atomic terms} are either variables, sorts, or terms of the form $\indic{n}\cdot x$.
\end{defi}

\begin{defi}[Domain of contexts]
The domain of a context is defined as follows:
\[
\begin{array}{l}
\dom([\,]) :=  \varnothing, \qquad \dom(\Ga, \indic{n}) := \dom(\Ga), \\
\dom(\Ga, x=t:A) := \dom(\Ga, x:A) := \dom(\Ga)\cup \single{x}.
\end{array}
\]
\end{defi}

\begin{remm}\hfill
\begin{enumerate}[(1)]
\item Vector notations are used instead of some sequences of expressions:
\[
\begin{array}{|@{\,\,\,\bullet\,\,}l@{\qquad}@{\bullet\,\,}l|}\hline
\multicolumn{2}{|l|}{}\\[-2ex]
\vec t  :=  t_1, ..., t_n &  f\, \vec t := f\, t_1\cdots t_n \\[1ex]
\Pi \vec x : \vec A . t := \Pi x_1:A_1.\, ... \, \Pi x_n:A_n .\, t &
\la \vec x : \vec A . t  :=  \la x_1:A_1.\, ... \, \la x_n:A_n .\, t \\[1ex]
\multicolumn{2}{|l|}{\bullet\,\,\oto{x / k :A := t} :=  x_0 / k_0 : A_0 := t_0, ...  , x_n / k_n : A_n := t_n}\\[0.3ex]
 \hline
\end{array}
\]

\item Note that we use two subscript styles. One is of the form $t_1, ..., t_n$, and the other is of the form $t_0, ..., t_n$, where $n$ is a natural number. The latter style will be used only in the definition of mutually recursive functions, i.e., in combination with $\textup{\textsf{fix}}$.\smallskip

\item Given a sequence $\vec \ell$, let $\lh(\vec \ell\,)$ denote its length.\smallskip

\item In the examples presented below, character strings are used instead of single character variables in order to emphasize the correspondence with real Coq-expressions.\smallskip

\item Given a declaration $\De$ and a variable $x$, let $\De(x) = A$ when $A$ is the only term such that $x : A$ occurs in $\De$.\smallskip

\item There are standard definitions of the sets of free variables in a context or a term, and of the substitution $t[x\bs u]$, where $t,\, u$ are terms and $x$ a variable. Formal definitions are given in Appendix \ref{A-defi}.\smallskip

\item Given a sequence $\de = x_1:t_1,...,x_n:t_n$ and a term $t$, let $t\single{\de} := t[x_1\bs\, t_1]\cdots [x_n\bs\, t_n]$ denote consecutive substitution. 
On the other hand, the simultaneous substitution of terms $t_1,...,t_n$ for $x_1,...,x_n$, respectively, in $t$ is denoted by $t[\de] := t[x_1\bs\,t_1,...,x_n\bs\, t_n]$.

\end{enumerate}
\end{remm}

\noindent To enable the reader to understand the intended meaning of terms and contexts, we explain some notations with examples. The examples will also be used in Section \ref{interpretation} to explain our model.

\begin{remm}\label{rem-treeforest}
  The expression $\indic{n}$ denotes a \emph{(mutually) inductive type}, and the subscript $n$ denotes the number of \emph{parameters}. $\De_I$ and $\De_C$ are two declarations containing inductive types and their constructors, respectively. The \emph{Parameters} are binders shared by all the constructors of the definition, and they are used to construct polymorphic types.
The parameters differ from other non-parametric binders in that the conclusion of each type of constructor invokes the inductive type with the same parameter values as its specification.
We refer to Lemma \ref{ind-fam} and Lemma \ref{ind-para}, which show the difference between parameters and non-parametric binders.\medskip

The mutual definition of trees and forests can be represented, for instance, by $\cald_{TF} = \indic{1}$, where
\begin{eqnarray*}
  \De_I & = & \ttt{tree} : \ttype_0\to\ttype_0,\, \ttt{forest} : \ttype_0\to\ttype_0\,,\\
  \De_C & = & \ttt{node} : \Pi (A:\ttype_0).\, A \to \ttt{forest}\, A \to \ttt{tree}\, A,\\
  & & \ttt{emptyf} : \Pi (A:\ttype_0).\,\ttt{forest}\,A,\\
  & & \ttt{consf} : \Pi (A:\ttype_0).\,\ttt{tree}\,A\to\ttt{forest}\,A\to \ttt{forest}\,A\,.
\end{eqnarray*}
The subscript $1$ implies that $(\texttt{\em A} : \ttype)$ is a parameter.
\medskip
\end{remm}

\begin{remm}\label{rem-nat}
If $\cald = \indic{n}$ and $x \in \dom(\De_I, \De_C)$, then $\cald\cdot x$ corresponds to the names of defined inductive types or their constructors.\medskip

The type for natural numbers and its two constructors can be represented by 
$\cald_N\cdot \tnat, \quad \cald_N\cdot \ttt{O}$, and $\cald_N\cdot \ttt{S}$,
respectively, where $\cald_N = \indic{0}, \De_I = \ttt{nat} : \ttt{Type}_0$, and $\De_C = \ttt{0}:\ttt{nat}, \ttt{S}:\ttt{nat} \to \ttt{nat}$\,. \medskip

In the examples presented below, however, we use character strings for better readability. Thus, for example, \ttt{nat}, \ttt{0}, and \ttt{S} are used instead of $\cald_N\cdot \tnat, \cald_N\cdot \ttt{O}$, and $\cald_N\cdot \ttt{S}$, respectively.
\end{remm}

\begin{remm}[\textsf{case} and \textsf{fix}]\label{size}
The term $\case{e,Q,\vec{h}\,}$ corresponds to the following Coq-expression
\begin{center}
\ttt{match $e$ as $y$ in $I\,\vec{\uscore}\,\vec u$ return $Q\, \vec u \, y$ with $\dots \mid C_i \, \vec{p}\, \vec v$ => $h_i \mid \dots$ end}
\end{center}
where
\begin{iteMize}{$\bullet$}
\item the term $e$ is of an inductive type $I\,\vec p\, \vec u$ for some terms $\vec p, \vec u$,
\item $\lh(\vec p) = \lh(\vec{\uscore})$,
\item $Q = \la\vec u : \vec U.\,\la y:I\,\vec p\,\vec u.\, Q'$ for some terms $\vec U, Q'$, 
\item the term $y$ is a fresh variable bound in $Q \, \vec u\, y$,
\item each $C_i$ is a constructor of type $\Pi\vec p : \vec P. \Pi \vec v : \vec V_i. I\, \vec p\,\vec w_i$ for some terms $\vec P, \vec{V_i}, \vec{w_i}$, and
\item each $h_i = \la \vec v : \vec V_i .h'_i$ for some term $h'_i$\,.
\end{iteMize}

The term $\fix{f_i}{\oto{f/k:A:=t}}$ denotes the $(i+1)$th function defined by a mutual recursion. The number $k_i$ denotes the position of the inductive binder on which recursion is performed for $f_i$. It corresponds to Coq's \ttt{struct} annotation used for the \equote{\emph{guarded}} condition in the termination check (cf. \cite{gimenez}).

\begin{enumerate}[(1)]
\item The addition function \ttt{plus} can be defined as follows:
\begin{center}
$\ttt{plus}=\fix{f} {f/ 1:\Pi (m, n:\ttt{nat}).\,\ttt{nat}:=\la (m, n:\ttt{nat}).\, \case{n,Q, h_0, h_1 }}$\,,
\end{center}
where $Q=\la (\ell:\ttt{nat}).\, \ttt{nat}$, $h_0 = m$, and $h_1 = \la (p:\ttt{nat}).\, \ttt{S}\, (f\, m\, p)$\,. 

\item The functions for measuring the size of trees and forests can be represented by 
$\ttt{Tsize} = \fix{g_0}{R}$ and $\ttt{Fsize} = \fix{g_1}{R}$,
 where $R=\oto{g/k:B := t}$, $k_0 = k_1 = 1$, and
\begin{eqnarray*}
B_0 & = & \Pi (A:\ttype_0).\, \Pi (t:\ttt{tree} \, A).\, \tnat\,, \\
B_1 & = & \Pi (A:\ttype_0).\, \Pi (f:\ttt{forest} \, A).\, \tnat\,, \\[1ex]
t_0 & = & \la (A:\ttype_0).\, \la (t:\ttt{tree} \, A).\, \case{t, Q_0, h_0}\,, \\
t_1 & = & \la (A:\ttype_0).\, \la (f:\ttt{forest} \, A).\, \case{f, Q_1, h_1, h_2}\,,\\
Q_0  &= & \la (t: \ttt{tree} \, A).\, \tnat\,,\\
Q_1  & = &  \la (f : \ttt{forest}\, A).\, \tnat\,,\\
h_0  & = &  \la (a:A). \la(f: \ttt{forest} \, A).\, \ttt{S}\, (g_1\, A \, f)\,,\\
h_1  & = &  \ttt{O}\,,\\
h_2  & = &  \la (t: \ttt{tree } A).\, \la (f:\ttt{forest } A).\ttt{plus}\,(g_0\, A\, t)\,(g_1\, A \, f).
\end{eqnarray*}
\end{enumerate}
\end{remm}

\subsection{Typing rules}
  The typing judgment $\Ga\edash M:A$ or $\Ga\edash M = N:A$ is defined simultaneously with the property $\wf(\Ga)$ of a \emph{well-formed valid context}  and the property $\Ga\vd M \prec N$ of cumulativity of types in Figures \ref{fig:basic} $\sim$ \ref{fig:constraint}. We provide short explanations of some rules. For a more detailed explanation, refer to \cite{coqman}, \cite{letouzey}, or \cite{pm-habili}.

\subsection*{Typing rules for basic terms and valid contexts (Figure \ref{fig:basic})}

\begin{figure}[h]
\begin{center}
\noindent \begin{tabular}{|p{14.5cm}|}\hline\\
 \hfill $\wf(\varnothing)$
 \hfill \rbm{\infer{\wf(\Ga,(x:A))}{\Ga\edash A:s & x\notin \dom(\Ga)}}
 \hfill \rbm{\infer{\wf(\Ga, (x:=t:A))}{\Ga\edash t:A & x\notin\dom (\Ga)}}
\hfill $(wf)$\\[5ex]

 \hfill \rbm{\infer{\Ga\edash \tprop:\ttype_i}{\wf(\Ga)}}
 \hfill \rbm{\infer{\Ga\edash \ttype_i:\ttype_j}{\wf(\Ga) & i<j}}
 \hfill $(ax)$\\[5ex]

 \hfill \rbm{\infer{\Ga\edash x:A}
   {\wf(\Ga) & (x:A)\in \Ga & \text{or} & (x:=t:A)\in \Ga} }
\hfill $(var)$\\[5ex]

 \hfill \rbm{\infer{\Ga\edash \lett{x:=t}{u}:U[x\bs t]}
   {\Ga, (x:=t :A)\edash u:U} }
\hfill $(let)$\\[5ex]

 \hfill \rbm{\infer{\Ga\edash (\lett{x:=t}{u})=(\lett{x:=t'}{u'}):U[x\bs t]}
   {\Ga\edash t=t':A & \Ga, (x:=t :A)\edash u=u':U} }
\hfill $(let\tminus eq)$\\[5ex]



\hfill \rbm{
  \infer
  {\Ga\edash \Pi x:A.B:s_3}
  {\Ga\edash A:s_1 & \Ga,x:A\edash B:s_2 & \calp(s_1,s_2,s_3)}
}\hfill $(\Pi)$\\[5ex]
 
 \hfill \rbm{\infer{\Ga\edash \Pi x:A.B= \Pi x:A'.B':s_3}
   {\Ga\edash A= A':s_1 & \Ga,x:A\edash B= B':s_2 & \calp(s_1,s_2,s_3)} }
\hfill $(\Pi\tminus eq)$\\[5ex]

\hfill \rbm{
  \infer
  {\Ga\edash \la x:A.M:\Pi x:A.B}
  {\Ga, x:A\edash M:B & \Ga\edash \Pi x:A.B:s}    
 } 
\hfill $(\la)$\\[5ex]


\hfill \rbm{
  \infer
  {\Ga\edash \la x:A.M= \la x:A'.M':\Pi x:A.B}
  { \Ga\edash A= A':s\quad \Ga, x:A\edash M= M':B \quad \Ga\edash \Pi x:A.B:s'}
}
\hfill $(\la\tminus eq)$\\[5ex]


\hfill \rbm{\infer{\Ga\edash MN: B[x\bs N]}
   {\Ga\edash M:\Pi x:A.B & \Ga\edash N:A} } 
\hfill $(app)$\\[5ex]

 \hfill \rbm{\infer{\Ga\edash MN= M'N': B[x\bs N]}
   {\Ga\edash M= M':\Pi x:A.B & \Ga\edash N= N':A} }
\hfill $(app\tminus eq)$\\[2ex]
\hline
\end{tabular}
\end{center}
\caption{Basic terms and valid contexts}
\label{fig:basic}
\end{figure}

Typing rules for standard constructions of $\la$- and $\Pi$-terms are given.\medskip

\noindent $(wf)$: Well-formed contexts contain well-typed terms, and they can be extended by well-typed inductive types, as in rule $(ind\tminus wf)$ of Figure \ref{fig:inductive}.\medskip

\noindent $(\Pi)$ and $(\Pi\tminus eq)$: $\calp(s_1,s_2,s_3)$ implies that
\begin{iteMize}{$\bullet$}
\item $s_2=s_3=\tprop$, or
\item $s_1\in\ttype_i,\, s_2=\ttype_j$ and $s_3=\ttype_k$ where $k\ge \max\{i,j\}$.
\end{iteMize}

\subsection*{Typing rules for inductive types and recursive functions (Figure \ref{fig:inductive})}

\begin{figure}[h]
\begin{center}
\begin{tabular}{|p{14.5cm}|}\hline\\[6ex]
  \hfill \rbm{
    \infer
    {\wf(\Ga,\indic{n})} 
    {\deduce
      {\Ga,\De_I\edash T:s_c\quad\text{for all $(c:T)\in \De_C$}}
      {\deduce
        {\Ga\edash A:s_d\quad\text{for all }(d:A)\in \De_I}
        {\cali_n(\De_I, \De_C) }
      }
    }
  }
\hfill $(ind\tminus wf)$\\[5ex]

\hfill \rbm{
  \infer
  {\Ga\edash \cald\cdot d : \De_I(d)[.\bs \cald]} 
  {\wf(\Ga) & \cald = \indic{n} \in \Ga & d\in\dom(\De_I)}
}
\hfill $(ind\tminus type)$\\[5ex]

\hfill \rbm{
  \infer
  {\Ga\edash \cald\cdot c : \De_C(c)[.\bs\cald]} 
  {\wf(\Ga) & \cald = \indic{n} \in \Ga & c\in\dom(\De_C)}
}
\hfill $(ind\tminus const)$\\[11.5ex]

\hfill \rbm{
  \infer
  {\Ga\edash\case{e,Q,(h_k)_{k}}: Q\, \vec u\, e}
  {\deduce
    {\Ga\edash h_k : \Pi\vec v : \vec V_k.\, Q\, \vec w_k\,(c_k\,\vec p\, \vec v) \text{ for all } (c_k : \Pi\vec p : \vec P.\, \Pi\vec v:\vec V_k.\, d_i\, \vec p\, \vec w_k)\in \De_C}
    {\deduce
      {\Ga\edash Q:B \quad \calc(d_i\,\vec{p}:A ; B) \quad \Ga\edash e : d_i \,\vec{p}\,\vec{u}}
      {\indic{n}\in\Ga & (d_i:\Pi \vec p : \vec P.\, A)\in \De_I & \lh(\vec p)=n }
    }
  }
}
\hfill $(case)$\\[11.5ex]

\hfill \rbm{
  \infer
  {\Ga\edash \case{e,Q,(h_k)_{k}}=\case{e',Q',(h'_k)_{k}}: Q\, \vec u\, e}
  {\deduce
    {\Ga\edash h_k=h'_k : \Pi\vec v : \vec V_k.\, Q\, \vec{w_k}\,(c_k\,\vec p\, \vec{v}) \text{ for all } (c_k : \Pi\vec p : \vec P.\, \Pi\vec v : \vec V_k.\, d_i\, \vec p\, \vec w_k)\in \De_C}
    {\deduce
      { \Ga\edash Q=Q':B \quad \calc(d_i\,\vec{p}:A ; B) \quad \Ga\edash e=e' : d_i \,\vec{p}\,\vec{u} }  
      {\indic{n}\in \Ga & (d_i:\Pi \vec p : \vec P.\, A)\in \De_I  & \lh(\vec p)=n }
    }
  }
}
\hfill $(case\tminus eq)$\\[5.5ex]

\hfill \rbmm{
  \infer
  {\Ga\edash \fix{f_j}{\oto{f / k : A := t}} : A_j}
  {\calf(\vec f, \vec A, \vec k, \vec t) & n=\lh(\vec k) & (\Ga\edash A_i: s_i)_{\forall i\le n}\quad (\Ga, \vec f : \vec A \edash t_i : A_i)_{\forall i\le n} & j\le n}
}
\hfill $(fix)$\\[10ex]

\hfill \rbmm{
  \infer
  {\Ga\edash \fix{f_j}{\oto{f / k : A := t}} = \fix{f_j}{\oto{f / k : A' := t'}} : A_j }
  {\deduce
    {(\Ga\edash A_i = A'_i: s_i)_{\forall i\le n}\quad (\Ga, \vec f : \vec A \edash t_i = t'_i: A_i)_{\forall i\le n}\quad j\le n} 
    {\calf(\vec f, \vec A, \vec k, \vec t\,) & \calf(\vec f, \vec{A'}, \vec k, \vec{t'}) & n=\lh(\vec k) }
  }
}
\hfill $(fix\tminus eq)$\\[3.5ex]
\hline
\end{tabular}
\end{center}
\caption{Inductive types and recursive functions}
\label{fig:inductive}
\end{figure}

\begin{figure}
\begin{center}
\begin{tabular}{|p{14.5cm}|}\hline\\
 \hfill \rbm{\infer{\Ga\edash M= M:A}{\Ga\edash M:A}}\hfill \rbm{\infer{\Ga\edash N= M:A}{\Ga\edash M=N:A}}\hfill $(ref)(sym)$\\[5ex]

 \hfill \rbm{\infer{\Ga\edash M= P:A}{\Ga\edash M= N:A & \Ga\edash N= P:A}}\hfill $(trans)$\\[5ex]

 \hfill \rbm{\infer{\Ga\edash M:B}{\Ga\edash M:A & \Ga\edash A= B:s}}\hfill \rbm{\infer{\Ga\edash M= N:B}{\Ga\edash M= N:A & \Ga\edash A= B:s}}\hfill $(conv)(conv\tminus eq)$\\[5ex]

 \hfill \rbm{
   \infer
   {\Ga\edash (\la x:A.M)N= M[x\bs N]: B[x\bs N]}
   {\Ga,x:A\edash M:B & \Ga\edash \Pi x:A.B:s & \Ga\edash N:A}
 }
 \hfill $(\be)$\\[5ex]


\hfill \rbm{\infer{\Ga\edash x = t:A}
{\wf(\Ga) & (x:=t:A)\in \Ga} }
\hfill $(\de)$\\[5ex]

\hfill \rbm{\infer{\Ga\edash (\lett{x:=t}{u}) = u[x\bs t]:U[x\bs t] }
{\Ga\edash t:A & \Ga, (x:=t :A)\edash u:U } }
\hfill $(\zeta)$\\[11ex]

\hfill \rbm{
  \infer
  {\Ga\edash\case{c_j\, \vec p\, \vec a, Q,(h_k)_{k}} = h_j\, \vec a: Q\, \vec{u}\, (c_j\, \vec p\, \vec a)}
  {\deduce
    {\Ga\edash h_k : \Pi\vec v : \vec V_k.\, Q\, \vec{w_k}\,(c_k\,\vec p\, \vec{v}) \text{ for all } (c_k : \Pi\vec p : \vec P.\, \Pi\vec v : \vec V_k.\, d_i\, \vec p\, \vec w_k)\in \De_C}
    {\deduce
      { \Ga\edash Q:B \quad \calc(d_i\,\vec{p}:A ; B) \quad \Ga\edash c_j\, \vec p\, \vec a : d_i \,\vec{p}\,\vec{u} }
      {\indic{n}\in\Ga & (d_i:\Pi \vec p : \vec P.\, A)\in \De_I & \lh(\vec p)=n }
    }
  }
}
\hfill $(\iota)$\\[9ex]

\hfill \rbm{
  \infer
  {\Ga\edash (\fix{f_j}{R} )\, \vec a = (t_j[f_i\bs (\fix{f_i}{R} )] \,\vec a) : A'_j\single{\vec x : \vec a} }
  {\deduce
    {R\equiv\oto{f / k : A := t}\quad A_j\equiv \Pi \vec x_j : \vec B_j.\, A'_j \quad \Ga\edash \vec a : \vec B_j\quad \lh(\vec B_j)=k_j+1 }
    {(\Ga\edash A_i: s_i)_{\forall i\le n}\quad (\Ga, \vec f : \vec A\edash t_i : A_i)_{\forall i\le n} & \calf(\vec f, \vec A, \vec k, \vec t) & j\le n} } }
\hfill $(\iota)$\\[7ex] 

\hfill  $\edash \tprop\prec \ttype_0$
\hfill $\edash \ttype_j\prec \ttype_{j+1}$\hfill ${(inc)}$ \\[2.5ex]

 \hfill \rbm{\infer{\Ga\edash M\prec P}{\Ga\edash M \prec N & \Ga\edash N\prec P}}\hfill $(trans\tminus inc)$ \\[5ex]
 
 \hfill \rbm{\infer{\Ga, x:C\edash A\prec B}{\Ga\edash A\prec B & \Ga\edash C:s & x\notin \dom(\Ga)}}\hfill $(weak\tminus inc)$\\[5ex]

 \hfill \rbm{\infer{\Ga\edash \Pi x:A_1.A_2\prec \Pi x:B_1.B_2}{\Ga\edash A_1 =B_1: s & \Ga, x:A_1\edash A_2\prec B_2}}\hfill $(\Pi\tminus inc)$ \\[5ex]

 \hfill \rbm{\infer{\Ga\edash N\prec P}{\Ga\edash M = N:s & \Ga\edash M\prec P}}\hfill \rbm{\infer{\Ga\edash P\prec N}
{\Ga\edash M = N:s & \Ga\edash P\prec M}}\hfill $(eq\tminus inc)$ \\[5ex]

 \hfill \rbm{\infer{\Ga\edash M:B}
{\Ga\edash M:A & \Ga\edash A\prec B} }
\hfill \rbm{\infer{\Ga\edash M= N:B}
{\Ga\edash M= N:A & \Ga\edash A\prec B } }
\hfill $(cum)(cum\tminus eq)$\\[2ex]
\hline
\end{tabular}
\end{center}
\caption{Judgmental equality and cumulative type universes}
\label{fig:equality}
\end{figure}

Typing rules for (mutually) inductive types, case distinctions, and (mutually) recursive functions are given.

\noindent $(ind\tminus wf)$: The positivity condition is crucial for defining an inductive type. A term $A$ is an \emph{arity ending in sort} $s$, $\arity(A,s)$, if it is convertible to $s$ or a product $\Pi x:A.B$, where $B$ is an arity ending in sort $s$. $A$ is called an \emph{arity}, $\arity(A)$, if $A$ is an arity ending in sort $s$ for some sort $s$.

A term $M$ satisfies the \emph{positivity condition} for a variable $x$ when
$M = \Pi \vec y : \vec A\,.\, x\,\vec u$ for some terms $\vec A, \vec u$ and
%
the variable $x$ occurs \emph{strictly positively} in $\vec A$.
A variable $x$ occurs \emph{strictly positively} in $M$ when
\begin{iteMize}{$\bullet$}
\item $x$ does not occur in $M$, or
\item $M\equiv \pifn{\vec y : \vec A}{(x\,\vec B)}$ and $x$ does not occur in $\vec A,\, \vec B$.\medskip
\end{iteMize}

Now, $\cali_n(\De_I, \De_C)$ represents the following conditions:
\begin{iteMize}{$\bullet$}
\item All the names contained in the domains of $\De_I$ and $\De_C$ must be mutually distinct and new.
    
\item All the types of $\De_I$ and $\De_C$ start with the same $n$ products, say, $\vec p : \vec P$.

\item Any occurrence of some $d \in \dom(\De_I)$ in $\De_C$ is of the form $(d\, \vec p\, \vec u)$, which is not applicable any more.
    
\item For all $d:A\in \De_I$, $A$ is an arity ending in sort $s_d$ such that $s_d \neq \tprop$. Thus, we do not use inductive definitions of type $\tprop$. Some propositions defined inductively can be constructed using an impredicative coding. See \cite{girard-proofs,coquand-meta}, and \cite{werner-set} for further details.
    
\item For all $c:T\in \De_C$, $T$ is the type of a constructor for an inductive type $d\in \dom(\De_I)$, i.e., $T$ is of the form $\Pi \vec p : \vec P.\,\Pi \vec z : \vec Z.\, (d\,\vec p\, \vec u)$. In this case, the sort $s_c$ in the third premise of the rule must be $s_d$.
    
\item $T$ satisfies the positivity condition for all $x \in \dom(\De_I)$.\smallskip
\end{iteMize}

\noindent\textbf{Notation.}
We use $\Ga \vd \indic{n}$ when all the premises of $(ind\tminus wf)$ are satisfied.\bigskip
    
\noindent $(ind\tminus type)$ and $(ind\tminus const)$: Given $\cald = \indic{n}$ and a term $A$, $A[.\bs\cald]$ implies that every occurrence of $z \in \dom(\De_I,\De_C)$ in $A$ is replaced with $\cald\cdot z$.\medskip

\noindent $(case)$ and $(case\tminus eq)$: $d_i$ and $c_k$ denote $\indic{n}\cdot d_i$ and $\indic{n}\cdot c_k$, respectively. Furthermore, $\lh(\vec u) = \lh(\vec w_k)$.\medskip

For an inductive type $d$ and an arity $B$, the relation $\calc(d\,\vec q :A ; B)$ is defined as follows:
    \begin{iteMize}{$\bullet$}
    \item $\calc(d\,\vec q:\tprop; d\,\vec q \to \tprop)$;
   \item $\calc(d\,\vec q : \tprop; d\,\vec q \to \ttype_j)$ iff $d$ is an inductive type that is empty or has only one constructor\footnote{This reflects the fact that no pattern matching is allowed on proof-terms, which would otherwise result in a paradox, as shown by \cite{coquand-meta}.}
    such that all the non-parametric arguments are of sort $\tprop$;
    \item $\calc(d\,\vec q : \ttype_j ; d\,\vec q\to s)$ for any sort $s$;
    \item $\calc(d\,\vec q : (\pifn{u:U}{A}) ; (\pifn{u:U}{B}))$ iff $\calc(d\,\vec q\, u:A;B)$.
    \end{iteMize}
This means that an object of the inductive type $d$ can be eliminated for proving a property $P$ of type $B$. Let $\calc(d\,\vec q ; B)$ denote $\calc(d\,\vec q : A ; B)$, where $A$ is the type of $d\,\vec q$.\medskip

\noindent $(fix)$ and $(fix\tminus eq)$:
$\calf(\vec f, \vec A, \vec k, \vec t)$ represents the following conditions:
\begin{iteMize}{$\bullet$}
\item $\lh(\vec f)=\lh(\vec A)=\lh(\vec k)=\lh(\vec t)$,


\item for each $t_i\in \vec t$, there is an inductive type $\cald\cdot d$, where $\cald=\indic{n}$, and a term $T_i$ such that 
  \begin{iteMize}{$-$}
  \item $t_i=\la \vec y : \vec Y.\,\la z:(\cald\cdot d)\,\vec u.\,t'_i$, where $\lh(\vec Y) = k_i$, and
    
  \item there is a \emph{constrained derivation} with respect to $z$ and $\De_I, \De_C$ such that
    \[
    \Ga^\epsilon, (\vec y : \vec Y)^\epsilon, z:^{=z} (\cald\cdot d)\,\vec u,(\vec f : \vec A)^{<\!z}_{\vec k} \edash t'_i:^\epsilon T_i^\epsilon
    \]
    where $(\vec f : \vec A)^{<\!z}_{\vec k}$ is the context composed of
    \[
    f_j:^\epsilon \Pi (\vec u : \vec B_j)^\epsilon . \Pi v:^{<\!z} X_j\,.\, P_j^\epsilon
    \]
    if $f_j: \Pi (\vec u : \vec B_j) . \Pi v: X_j\,.\, P_j$ is from  $\vec f : \vec A$ and $\lh(B_j) = k_j$.\medskip
  \end{iteMize}
  
  The condition for \emph{constrained derivation} ensures that the constructed terms are normalizing terms. A formal definition is given in Appendix \ref{constraint}. Informally, it means that $t_i$ can only contain decreasing recursive calls: if $f_j$ appears in $t_i$, then it must have at least $k_j+1$ arguments, and its $(k_j+1)$th argument must be structurally smaller than the initial inductive argument $z$ $($Thus, any subterm of an inductive term obtained by going through at least one constructor is structurally smaller than the initial term.$)$.
\end{iteMize}

\subsection*{Judgmental equality and type universes (Figure \ref{fig:equality})}
The rules in Figure \ref{fig:equality} stipulate that the judgmental equality based on reductions is an equivalence relation. De Bruijn's {\it telescope} notation is very useful: $\Ga \vd \vec t : \vec A$ with $\lh(t) = n$ implies that
\begin{iteMize}{$\bullet$}
\item $\Ga, x_1:A_1, ..., x_{j-1}:A_{j-1} \vd A_j : s_j$ for all $j\in \single{1, ..., n}$, and

\item $\Ga \vd t_j : A_j[x_1\bs t_1] \cdots [x_{j-1} \bs t_{j-1}]$ for all $j\in \single{1, ..., n}$.
\end{iteMize}

  
\section{Set-theoretic model construction}\label{interpretation}

\subsection{Background}
We must resolve a dilemma related to the construction of a set-theoretic model of CC and its extensions. In a proof-irrelevant model, each type expression should have an obvious set-theoretic interpretation; however, it is well known that impredicative or polymorphic types, such as $\tprop$, can only have a trivial set-theoretic interpretation, as shown by \cite{reynolds}. Hence, it is necessary to assign a singleton or the empty set to each term of type $\tprop$. 

In constructing a set-theoretic model of Coquand's CC, \cite{miwe} provided the following solution. Under the assumption of the existence of a urelement $\bullet$ that does not belong to the standard universe of set theory, the sort $\tprop$ is associated with $\single{\varnothing, \single{\bullet}}$. Furthermore, the application and $\la$-abstraction terms are interpreted by $\app$ and $\lam$, respectively, which are defined as follows:
\begin{eqnarray*}
  \app(u,x) & := &
  \begin{cases}
    \bullet & \text{if } u = \bullet\,, \\
    u(x) & \text{otherwise.}
  \end{cases}\\
  \lam(f) & := &
  \begin{cases}
    \bullet & \text{if } f(x) = \bullet \text{ for all } x\in \dom(f)\,, \\
    f & \text{otherwise.}
  \end{cases}\\
\end{eqnarray*}

\begin{rem}\label{counter-subtyping}
This construction does not correctly model the cumulative relation between $\tprop$ and $\ttype_i$, as demonstrated in the following example; \cite{werner-proof} showed that it can be easily extended to the cumulative type universes when the subtyping relation between $\tprop$ and $\ttype_i$ does not exist\smallskip

Consider $I = \la (A:\ttype_0).\, A \to A$. Then, its type of set-theoretic interpretation is not deterministic. Suppose that $P$ is a true proposition. Then, $\lrbk{I\, P}$ depends on the type we have assigned to $P$, that is, $\tprop$ or $\ttype_0$. In the former case, $\lrbk{I\, P} = \single{\bullet}$ since $P \to P $ is a tautology, whereas in the latter case, $\lrbk{I \, P} = \lrbk{I}(\lrbk{P}) = \menge{f}{f : \single{\bullet} \to \single{\bullet}} \neq \single{\bullet}$ since $\lrbk{P} = \single{\bullet}$. 
\end{rem}

Another solution was provided by \cite{aczel-relating}. He used the trace encoding of functions in order to provide an adequate interpretation of the impredicative type $\tprop$ of propositions and its relationship with $\ttype_i$. For this reason, we adopt Aczel's solution.

\begin{defi}[Trace encoding of set-theoretic functions]
Let $u, x, f$ denote sets. Then,
\begin{eqnarray*}
  \app(u,x) & := & \menge{z}{(x,z)\in u }, \\
  \lam(f) & := & \bigcup_{(x,y)\in f}(\{x\} \times y) \,\,\,\, =\,\,\, \bigcup_{(x,y)\in f} \menge{(x,z)}{z\in y}.
\end{eqnarray*}
\end{defi}

Note that for any function $f$ and any $x\in \dom(f)$, we have
\[
\app(\lam(f),x)
= \menge{y}{(x,y)\in \lam(f)}
=  \menge{y}{y\in f(x)}
= f(x)\,.
\]

\noindent\textbf{Notations.}
\begin{enumerate}[(1)]
\item In the remainder of this paper, $\downarrow$ is used if something is well defined, and $\uparrow$ is used otherwise.
\item Given sets $A, B(x)$, $x \in A$, let $\prod_{x \in A} B(x)$ denote the set of all functions $f$ such that $\dom(f) = A$ and $f(x) \in B(x)$ for all $x \in A$.

\item Given a function $f \in \prod_{x_1\in A_1} \cdots \prod_{x_n \in A_n(x_1,...,x_{n-1})}\, B(x_1,...,x_n)$, we use the notation $\vec\lam_n (f)$ (resp. and $\vec{\app}(f, \vec x)$) for the $n$-times application of $\lam$ (resp. $\app$):
\begin{eqnarray*}
\vec\lam_n (f) & := & \menge{(x_1,...,x_n,y)}{x_1 \in A_1, ..., x_{n} \in A_{n}(x_1, ..., x_{n-1}), y \in f(x_1, ..., x_n)}\\
\vec\app_n (f, \vec x\, ) & := & \app(... (\app(\app(f, x_1), x_2),...), x_n)
\end{eqnarray*}
\end{enumerate}

We suppress the subscript $n$ when the number of times we want to apply $\lam$ or $\app$ is obvious from the context. Note that $\vec\lam_0(f) = f$ and $\vec\app_0(f, nil) = f$.
\begin{lem}[\cite{aczel-relating}]\label{A-poly}
Given a set $A$, assume $B(x)\tm 1$ for all $x \in A$.
\begin{enumerate}[\em(1)]
\item $\menge{\lam(f)}{f \in \prod_{x\in A} B(x)}\tm 1$.
\item $\menge{\lam(f)}{f \in \prod_{x\in A} B(x)} = 1$ iff\, $\forall x\in A\, (B(x)=1)$.
\end{enumerate}
\end{lem}
\begin{proof}
Let $f\in \prod_{x\in A} B(x)$, i.e.,
  \[
  f =
  \begin{cases}
    \uparrow & \text{ if } \exists x\in A\, (B(x)=\varnothing), \\
    \menge{(x,\varnothing)}{x\in A}  & \text{ otherwise.}
  \end{cases}
  \]
  Then, we have
  \[
  \lam(f) =
  \begin{cases}
    \uparrow & \text{ if } \exists x\in A\, (B(x)=\varnothing),\\
    \varnothing   & \text{ otherwise.}
  \end{cases}
  \]
  This also implies that $\menge{\lam(f)}{f\in \prod_{x\in A} B(x)} = 1$ iff $\forall x\in A\, (B(x) = 1)$.
\end{proof}

\begin{rem}
A useful feature of trace encoding is that $\app(u,a)$ and $\lam(f)$ are always defined for any sets $u, a, f$, including the empty set. This, however, implies that we sometimes lose the information of the domain of a given function $f$, i.e., we cannot trace back to $\dom(f)$ starting from $\lam(f)$. We will see that the use of judgmental equality enables us to avoid such a loss when only well-typed terms are involved.
\end{rem}

\subsection{Inductive types and rule sets}

Here, we follow the approaches of \cite{aczel-relating} and \cite{dybjer} for the construction of a set-theoretic interpretation of inductive types. We are particularly interested in \emph{rule sets}.

We are going to work on the basis of ZF set theory with an axiom guaranteeing the existence of countably many (strongly) inaccessible cardinals. Note that such an axiom is independent of ZFC. \cite{werner-set} showed that such an axiom is sufficient for a set-theoretic interpretation of the cumulative type universes and predicatively inductive types.
However, it is not clear whether this axiom is necessary for our construction. Indeed, the required feature of an inaccessible cardinal $\ka$ is the closure property of the universe $\calv_\ka$ under the powerset operation. This is a necessary condition for the interpretation of inductive types.

Henceforth, assume that there are countably many (strongly) inaccessible cardinals. Let $\ka_0=\omega$ and $\ka_1, \ka_2$,  ... enumerate these inaccessible cardinals. We associate each sort $\ttype_i$ with its $rank(\ttype_i) := \ka_i$. If $(\calv_\al)_{\al \in Ord}$ denotes the (standard) universe of sets defined as follows, then $\calv_\ka$ is a model of ZF:
\[
\calv_0 :=  \varnothing \quad \text{and} \quad
\calv_\al := \bigcup_{\be \in \al} \calp(\calv_\be) \quad \text{if } \al > 0
\]
$Ord$ denotes the class of all ordinals, $\la$ denotes a limit ordinal, and $\calp$ denotes the power set operator. In particular, if $\ka$ is an inaccessible cardinal, $A \in \calv_\ka$, and for every $a \in A$, $B_a \in \calv_\ka$, then, $\prod_{a \in A} B_a \in \calv_\ka$. Let $rank(\tprop) := -1$ and $\calv_{{-1}} = \single{0,1}$ for convenience.  Refer to \cite{drake} for further details about inaccessible cardinals.

A \emph{rule} on a base set $U$ is a pair of sets $\cons{u,v}$, often written as $\frac u v$, such that $u\tm U$ and $v\in U$. A set of rules on $U$ is called a \emph{rule set} on $U$. Given a rule set $\Phi$ on $U$, a set $w\tm U$ is $\Phi$-closed if for any $\frac u v \in \Phi$, $v\in w$ whenever $u\tm w$. Note that there is the least $\Phi$-closed set
\[
\cali(\Phi):=\bigcap\menge{w\tm U}{w\,\,\,\text{$\Phi$-closed}}\,.
\]
In fact, it is well known that each rule set $\Phi$ on $U$ generates a monotone operator on $\calp(U)$
\[
\Ga_\Phi(X) := \menge{v\in U}{\text{there exists some $u\tm X$ such that $\frac u v\in \Phi$}}
\]
such that $\cali(\Phi)$ is the least fixed point of $\Ga_\Phi$. Assuming that $\Phi$ is a rule set on $I\times U$, $\Phi$ defines a family $\calif(\Phi)$ of sets in $U$ over $I$ as
\[
\calif(\Phi)(i) :=\menge{u\in U}{\cons{i,u}\in \cali(\Phi)}
\]
for each $i\in I$. 

A rule set is \emph{deterministic} provided that it contains at most one rule with a given conclusion. The rule sets  defined below by an inductive definition are deterministic. This makes it possible to interpret functions defined by structural recursion on a certain inductive type as set-theoretic functions. The interpretations are defined on the corresponding set-theoretic inductively defined set, which is the fixpoint of a monotone operator. Refer to \cite{aczel-induction} and \cite{moschovakis74,moschovakis} for further details about rule sets, monotone operators, and fixpoints.

Below, we describe the interpretations of inductive and recursive types with some examples. Given a well-defined (mutually) inductive type $\indic{n}$, where
\[
  \De_I = x_0:A_1,..., x_\ell:A_\ell \quad\text{and}\quad  A_i  =  \Pi\vec p : \vec P.\, \Pi\vec a_i : \vec B_i .\, s_i\,,
\]
let $rank(x_i) := rank(s_i)$. \bigskip

\noindent\textbf{Notations.}
\begin{enumerate}[(1)]
\item With each context $\Ga$, we associate a set $\lrbk{\Ga}$ of $\Ga$-{\it valuations} of the form $\cons{\al_1,\cdots, \al_n}$, where $n$ is the length of $\Ga$ and $\cons{,...,}$ denotes a sequence of  a finite length. Given a sequence $L=\cons{\al_1,\cdots, \al_n}$ and a natural number $i < n$, we set
$(L)_i = \al_{i+1}$. If $\al_{i+1}$ itself is a sequence of length $m$, then we write $(L)_{i,j}$ for $(\al_{i+1})_j$ if $j < m$, etc.

\item $\al, \be, \al_i, \be_i$ vary over single values while $\ga, \de, \ga_i, \de_i$ vary over valuations. $nil$ denotes the empty sequence. Given two valuations $\ga$ and  $\de$, the notation $\ga,\de$ denotes their concatenation. If $\de = \cons{\al}$, then we write $\ga, \al$ instead of $\ga, \cons{\al}$.

\item With each pair $(\Ga,t)$ formed by a context $\Ga$ and a term $t$, we associate a function $\lrbk{\Ga\vdash t}$  that is partially defined on $\Ga$-valuations: $\lrbk{\Ga\vdash t}_\ga$ denotes $\lrbk{\Ga\vdash t}(\ga)$ when $\ga \in \lrbk{\Ga}$.

\item In the following, we write $\lrbk{t}$ for $\lrbk{\Ga \vd t}_\ga$ if $\Ga$ and $\ga \in \lrbk{\Ga}$ are fixed in the context. Similarly, we use the notation $\vec u \in \lrbk{\vec A\,}$ for $u_1 \in \lrbk{\Ga \vd A_1}_\ga$, ..., $u_n \in \lrbk{\Ga, x_1:A_1, ... x_{n-1}: A_{n-1} \vd A_n)}_{\ga, u_1 ... u_{n-1}}$ for some context $\Ga$ and $\Ga$-valuation $\ga \in \lrbk{\Ga}$.
\end{enumerate}

\subsection{Interpretation of inductive types}\label{int-ind}

Here, we claim the existence of the interpretations of inductive types that satisfy the soundness of the rules $(ind\tminus wf), (ind\tminus type),$ and $(ind\tminus const)$ when the conditions in the typing rules are fullfilled. The formal definition is given in Appendix \ref{A-inductive}. Refer to \cite{dybjer}, whose idea is generalized in this paper.

\begin{lem}
  Suppose $\Ga \vd \cald$, where $\cald= \indic{n}$. Let $\ga \in \lrbk{\Ga}$ be given. As mentioned before, we suppress $\Ga$ and $\ga$ for better readability. Further suppose that 
\[
\begin{array}{c}
  \De_I := d_0:A_0,..., d_\ell:A_\ell\,,\,\,
  \De_C := c_1:T_1,..., c_m:T_m\,, \\[2ex]
  A_i  :=  \Pi\vec p : \vec P.\, \Pi\vec b_i : \vec B_i .\, s_i\,,\,\,
  T_k  := \Pi\vec p : \vec P.\, \Pi\vec z_k : \vec Z_k.\, d_{i_{k}}\, \vec p\, \vec t_k \,.
\end{array}
\]
Then, there is some rule set $\Phi$ such that the following interpretation of $\cald\cdot d_i$ and $\cald\cdot c_k$ satisfies the soundness of the rules $(ind\tminus type)$ and $(ind\tminus const)$:
\begin{iteMize}{$\bullet$}
\item $\lrbk{\cald\cdot d_i}:= \vec\lam(f_i)$
 where 
 $f_i (\vec p, \vec b_i) := \calif(\Phi)(i,\vec p, \vec b_i)$
for $\vec p,\vec b_i:\lrbk{\vec P, \vec B_i}$\,,
\item $\lrbk{\cald\cdot c_k}:= \vec\lam(g_k)$
 where
  $g_k(\vec p, \vec z_k) := \cons{k,\vec z_k}$
for $\vec p, \vec z_k : \lrbk{\vec P, \vec Z_k}$\,. 
\end{iteMize}
\end{lem}

\begin{remm}\label{rem-inductive}
  The positivity condition is crucial for showing that the construction of the rule set $\Phi$ in Appendix \ref{A-inductive} is well defined.
\end{remm}

The elements of our rule sets $\Phi$ are of the form
$
\frac{u}{\cons{i, \vec p, \vec t, \cons{j, \vec v\,}}},
$
where $i$ denotes the $i$th inductive type $d_i$, 
$\vec p$ denote the parameters, 
$\vec t$ denote the non-parametric arguments of $d_i$, 
$j$ denotes the $j$th constructor of $d_i$, and 
$\vec v$ denote the non-parametric arguments of the $j$th constructor. Note that $\vec p, \vec t, \vec v$ could be empty.

\begin{bsp}[Natural numbers] Let $\cald_N$ be the inductive type for natural numbers, as in Remark \ref{rem-nat}. 
Then,
\begin{eqnarray*}
  \displaystyle \Phi_{\ttt{nat}}  & = & 
  \Bigg \{ \frac{\varnothing}{\cons{0, \cons{1} } }\Bigg \}\,\,\cup\,\, 
  \Bigg \{ \frac{\single{v } }{\cons{0, \cons{2,v} } }
  \,\, \Big\lvert\,\, 
  v\in \calv_{\ka_0} \Bigg \}\,,
\end{eqnarray*}
$\lrbk{\ttt{nat}} = \calif(\Phi_{\ttt{nat}})(0)$,
$\lrbk{\ttt{O}} = \cons{1}$, and
$\app(\lrbk{\ttt{S}}, n) = \cons{2,n}$ for any $n\in \lrbk{\ttt{nat}}$.
\end{bsp}

\begin{bsp}[Inductive families]\label{ind-fam} The following Coq-expression shows a typical use of inductive families.

\begin{verbatim}
  Inductive toto : Type -> Type :=
   | Y1 : forall x : Type, toto x
   | Y2 : forall x : Type, toto nat -> toto x -> toto x.
\end{verbatim}
  The inductive type $\ttt{toto}$ can be represented by $\cald_{toto}=\ind_0\single{\De_I:=\De_C}$, where
\begin{eqnarray*}
  \De_I & := & \ttt{toto}\,:\ttype_1\to \ttype_1\,,\\
  \De_C & := & \ttt{Y}_1: \Pi x:\ttype_1.\, \ttt{toto}\,\, x, \quad \ttt{Y}_2: \Pi x:\ttype_1.\, \ttt{toto}\,\, \tnat \to \ttt{toto}\,\, x\to \ttt{toto}\,\,\, x\,.
\end{eqnarray*}
Then,
\begin{equation*}
  \Phi_{\ttt{toto}} =
  \Bigg \{ 
  \frac{\varnothing}{\cons{0, x,\cons{1,x} } } 
  \,\,\Big\lvert \,\, x\in\calv_{\ka_1} \Bigg \}
  \,\, \cup \,\,
  \Bigg \{ \frac{\single{\cons{0, \lrbk{\tnat}, v_1}, \cons{0, x, v_2} } }
                       {\cons{0, x,\cons{2,x, v_1, v_2 } } }
    \,\, \Big\lvert\,\, 
  x,  v_1, v_2\in \calv_{\ka_1} \Bigg \}\,,
\end{equation*}
\vspace{.5ex}
$\app(\lrbk{\ttt{toto}}, x) = \calif(\Phi_{\ttt{toto}})(0,x)$, 
$\app(\lrbk{\ttt{Y}_1}, x)  =  \cons{1,x}$, and 
$\app(\lrbk{\ttt{Y}_2}, x, a, b)  =  \cons{2,x,a,b}$,
where $x\in \calv_{\ka_1}$, $a\in \app(\lrbk{\ttt{toto}}, \lrbk{\tnat})$, and $b\in \app(\lrbk{\ttt{toto}}, x)$.
\end{bsp}

\begin{bsp}[Inductive types with parameters]\label{ind-para}
  The following Coq-expression shows a typical use of parametric inductive types.

\begin{verbatim}
     Inductive titi (x : Type) : Type :=
      | Z1 : titi x
      | Z2 : titi nat -> titi x -> titi x.
\end{verbatim}
The inductive type $\ttt{titi}$ can be represented by $\cald_{titi}=\indic{1}$, where
\begin{eqnarray*}
  \De_I & := & \ttt{titi}\,:\ttype_1\to \ttype_0\,,\\
  \De_C & := & \ttt{Z}_1: \Pi x:\ttype_1.\, \ttt{titi}\,\, x, \quad  \ttt{Z}_2: \Pi x:\ttype_1.\, \ttt{titi}\,\, \tnat \to \ttt{titi}\,\, x\to \ttt{titi}\,\, x\,.
\end{eqnarray*}
Then, 
\begin{eqnarray*}
  \Phi_{\ttt{titi}} = 
  \Bigg \{ 
  \frac{\varnothing}{\cons{0, x,\cons{1}} } 
  \,\,\Big\lvert \,\, x\in\calv_{\ka_1} \Bigg \}\,\, \cup  \,\,\Bigg \{ \frac{\single{\cons{0, \lrbk{\tnat}, v_1}, \cons{0, x, v_2} } }
    {\cons{0, x,\cons{2, v_1, v_2 } } }
    \,\, \Big\lvert\,\, 
  x\in\calv_{\ka_1},\,\, v_1, v_2\in \calv_{\ka_0} \Bigg \}\,,
\end{eqnarray*}
$\app(\lrbk{\ttt{titi}}, x)  =  \calif(\Phi_{\ttt{titi}})(0,x)$,
$\app(\lrbk{\ttt{Z}_1},x)  :=  \cons{1}$, and
$\app(\lrbk{\ttt{Z}_2},x,a,b)  :=  \cons{2,a,b}$,
where $x\in \calv_{\ka_1}$, $a\in \app(\lrbk{\ttt{titi}}, \lrbk{\tnat})$, and $b\in \app(\lrbk{\ttt{titi}}, x)$.
\end{bsp}

\begin{remm}
  Note that in Coq, \ttt{toto} cannot have \ttt{Type -> Set} as its type, unlike \ttt{titi}. This difference is also reflected in their interpretations.
\end{remm}

\begin{bsp}[Mutually inductive types with parameters]
  Two inductive types $\ttt{tree}$ and $\ttt{forest}$ defined by $\cald_{TF}$ in Remark \ref{rem-treeforest} can be interpreted by means of
  the following rule set:
\begin{eqnarray*}
  \displaystyle \Phi & := & 
  \Bigg \{ \frac{\single{\cons{1, A, v_1} } }{\cons{0, A, \cons{1,a, v_1 } } }
  \,\, \Big\lvert\,\,  A, v_1\in \calv_{\ka_0}, a\in A \Bigg \}\\[1ex]
  & & \cup\,\, \Bigg \{\frac{\varnothing}{\cons{1, A, \cons{2} } } \,\, \Big\lvert\,\, A \in \calv_{\ka_0} \Bigg \}\,\,
  \cup\,\, \Bigg \{ \frac{\single{\cons{0, A, v_1},\cons{1, A, v_2} } }{\cons{1, A, \cons{3,v_1,v_2} } }
  \,\, \Big\lvert\,\, 
  A, v_1, v_2\in \calv_{\ka_0} \Bigg \}\,.
\end{eqnarray*}
\end{bsp}

\subsection{Interpretation of well-founded structured recursion}\label{int-rec}

A set defined by a (mutual) induction generates a canonical well-founded relation on the set, i.e., the relation defined according to the inductive construction of the elements, the so-called {\em structurally-smaller-than-relation}. This is the basis for the discipline of structural recursion, which stipulates that recursive calls consume structurally smaller data. 

Here, we claim the existence of the interpretations of recursive types that satisfy the soundness of the rules $(fix)$, $(fix\tminus eq)$, and $(\iota)$. A formal definition  is given in Appendix \ref{A-recursion}. Refer to \cite{dybjer}, whose study provides the basic idea.

\begin{lem}
Suppose $\Ga\edash \fix{f_\ell}{R} : A_j$, where
\[
\begin{array}{c}
R = \oto{f / k : A := t},\quad A_i\equiv \Pi \vec x_i : \vec B_i.\, A'_i, \quad \lh(\vec B_i)=k_i+1, \quad \ell\le n, \\[2ex]
(\Ga\edash A_i: s_i)_{\forall i\le n}, \quad (\Ga, \vec f : \vec A \edash t_i : A_i)_{\forall i\le n}, \quad \calf(\vec f, \vec A, \vec k, \vec t) \,.
\end{array}
\]
Let $\ga \in \lrbk{\Ga}$ be given. We suppress $\Ga$ and $\ga$ for better readability. Then, there is a rule set $\Psi$ such that the following interpretation of $\fix{f_\ell}{R}$ satisfies the soundness of the rules $(fix)$, $(fix\tminus eq)$, and $(\iota)$:
\begin{iteMize}{$\bullet$}
\item $\lrbk{\fix{f_\ell}{R}} = \vec \lam (h)$, where
$h(a_1,...,a_{k_\ell},\cons{k,\vec z_k}) = \calif (\Psi) (a_1,...,a_{k_\ell},\cons{k,\vec z_k})$ 
for $\vec a, \cons{k,\vec z_k} \in \lrbk{\vec B_\ell}$.
\end{iteMize}
\end{lem}

\begin{remm}\label{rem-recursion}
  The condition for constrained derivation is essential. Indeed, constrained derivation corresponds to guarded recursion  defined by \cite{gimenez}; hence, it guarantees that the construction of the rule set $\Psi$ in Appendix \ref{A-recursion} is well defined.
\end{remm}

The elements of our rule sets $\Psi$ are of the form
$
\frac{u}{\cons{\vec x, \cons{k,\vec y\,}, b}},
$
where $k$ denotes the $k$th constructor of the inductive type $d$ on which the recursion is performed, $\vec y$ denote the non-parametric arguments of $d$, $\vec x$ denotes the list of rest arguments of the constructor, and $b$ is the result of the function. Note that $\vec x, \vec y$ could be empty.




\begin{bsp}[Primitive recursion] The Coq-expression stated below is a general form of primitive recursion.
\begin{verbatim}
     Fixpoint PRec (A:Type)(g:A)(h:nat -> A -> A)(n:nat) {struct n} : A :=
       match n with
         | O => g
         | S p => h p (PRec A g h p)
       end.
\end{verbatim}
The corresponding term is $\ttt{PRec} := \fix{f_0}{f_0/3: B := t}$, where 
\begin{eqnarray*}
B &  = & \Pi (A:\ttype_i).\, \Pi (g:A).\, \Pi (h:\tnat\to A \to A).\, \Pi (n:\tnat).\,\tnat\,, \\
t & = & \la (A:\ttype_i).\, \la (g:A).\, \la (h:\tnat\to A \to A).\, \la (n:\tnat).\, \case{n,P,h_1, h_2}\,,
\end{eqnarray*}
$P = \la (\ell:\tnat).\, A$, $h_1 = g$, and $h_2 = \la (p:\tnat).\, h\, p\, (f_0\, A\, g\, h\, p)$. $\lrbk{\ttt{PRec}}$ is characterized by the following rule set $\Psi_{\ttt{PRec}} :=$
\begin{eqnarray*}
  & &\Bigg \{ \frac{\varnothing}{\cons{A, g, h,\cons{1}, g}}\,\,\Big\lvert\,\, A\in \calv_{\ka_i}, g \in A, h\in \lrbk{\tnat \to A \to A} \Bigg \}\\
& & \,\,\cup\,\, 
\Bigg \{ \frac{\single{\cons{A, g, h,p, v}}}{\cons{A, g, h, \cons{2,p}, \vec \app(h, p, v) }}\,\, \Big\lvert \,\, A\in \calv_{\ka_i}, g \in A, h\in \lrbk{\tnat \to A \to A}, p \in \lrbk{\tnat}, v\in A \Bigg \}\,.
\end{eqnarray*}

Given a type $A$, $\lrbk{\ttt{PRec}\, A}$ denotes the primitive recursor with values from $A$. For instance, 
\[
\vec \app(\lrbk{\ttt{plus}}, m, n )= \calif(\ttt{Prec})(\lrbk{\tnat}, m, \lrbk{h}, n)
\]
where
$h = \la (p:\tnat).\, \la (\ell:\tnat).\, \ttt{S}\, \ell $ and $m, n \in \lrbk{\ttt{nat}}$.

\end{bsp}

\begin{bsp}[Mutually recursive functions] The interpretations of $\ttt{Tsize}$ and $\ttt{Fsize}$ from Remark \ref{size} are characterized by the following rule set $\Psi_{\ttt{Size}} :=$
\begin{eqnarray*}
  & & \Bigg \{ \frac{\single{\cons{A, f, v'}} }
  {\cons{A, \cons{1, a, f}, \lrbk{\ttt{S}\, v'}} }\,\,
  \Big\lvert\,\, 
  A \in \calv_{\ka_0}, f\in \app(\lrbk{\ttt{forest}}, A), v' \in \lrbk{\tnat}, a\in A \Bigg \}\\
  & & \,\,\cup\,\, 
  \Bigg \{ \frac{\varnothing }
  {\cons{A, \cons{2}, \lrbk{\ttt{O}}} }\,\,
  \Big\lvert\,\, 
  A \in \calv_{\ka_0} \Bigg \}\\
  & & \,\,\cup\,\, 
  \Bigg \{ \frac{\single{\cons{A, t, v_1' }, \cons{A, f', v_2' } }}
  {\cons{A, \cons{3, t, f'}, \app(\lrbk{\ttt{plus}}, v_1', v_2')} }\,\,\Big\lvert \,\,  A\in \calv_{\ka_0}, t \in \app(\lrbk{\ttt{tree}}, A), \\
  & & \hspace{6.7cm} f' \in \app(\lrbk{\ttt{forest}}, A), v_1', v_2' \in \lrbk{\tnat}\Bigg \}\,.
\end{eqnarray*}

\end{bsp}

\section{Set-theoretic model and soundness}\label{soundness}
Since the denotations $\lrbk{\Ga}$ and $\lrbk{\Ga\vd t}$ will be defined by mutual induction on the size of their arguments, we need a size function $|\cdot|$ that guarantees the termination. In particular, the following properties should be satisfied:
\begin{iteMize}{$\bullet$}
\item $|\Ga| < |\Ga \vd A| < | \Ga, x:A |$,

\item $|\Ga \vd t |, | \Ga \vd A | < |\Ga, x:=t : A|$,


\item $|\De_I(d)|$, $|\De_C(c)| < |\indic{n}\cdot x|$ for all $x \in \dom(\De_I, \De_C)$, $d \in \dom(\De_I)$, and $c \in \dom(\De_C)$,

\item $|A_i|$, $|t_j| < |\oto{f/k:A:=t}|$ for all $A_i \in \vec A$ and $t_j \in \vec t$.
\end{iteMize}
An adequate size function can be defined by a simple extension of the one defined by \cite{miwe}:
 \begin{iteMize}{$\bullet$}
 \item The size $|t|$ of a term $t$ is (recursively) defined as the sum of the sizes of its immediate subterms plus $1$. 
 \begin{iteMize}{$-$}
\item The immediate subterms of $\indic{n}\cdot x$ are $\De_I(d)$ and $\De_C(c)$, where $d \in \dom(\De_I)$ and $c \in \dom(\De_C)$.

\item The immediate subterms of $\fix{f_j}{\oto{f/k:A:=t}}$ are $\vec A, \vec t$.\smallskip
 \end{iteMize}

 \item The size of a context $\Ga$ is defined as follows:
   \begin{iteMize}{$-$}
   \item $|[\,]| = 1/2$,
   \item $|\Ga, (x:t)| = |\Ga| + |t|$,
   \item $|\Ga, (x:=t:A)| = |\Ga| + |t| + |A|$,
   \item $|\Ga, \indic{n}| = |\Ga| +1$.\smallskip
   \end{iteMize}
   
   \item $|\Ga \vd t| = |\Ga| + |t| - \frac 1 2$.
 \end{iteMize}
 
 \begin{remm}
 As mentioned in Remark \ref{rem-inductive} and Remark \ref{rem-recursion}, the positivity condition and the condition for constrained derivation play a crucial role for establishing the soundness proof of Theorem \ref{thm-soundness}.
 \end{remm}


\begin{defi}\label{def-interpretation}
 The set-theoretic interpretations of $\lrbk{\Ga}$ and $\lrbk{\ga\vd t}$ are defined by a mutual induction on the size of their arguments.
  \begin{enumerate}[(1)]
  \item For each context $\Ga$, the set $\lrbk{\Ga}$ is defined as follows:
\begin{align*}
  \lrbk{[\,]} & :=  \{\texttt{nil}\}, \\[.3ex]
  \lrbk{\Ga, x:A} & :=  \menge{\ga,\al}{\ga\in\lrbk{\Ga}, \lrbk{\Ga\vd A}_\ga\downarrow \text{ and } \al\in\lrbk{\Ga\vd A}_{\ga}}, \\[.3ex]
  \lrbk{\Ga, x:=t:A} & :=  \menge{\ga,\al}{\ga\in\lrbk{\Ga},\, \lrbk{\Ga\vd A}_\ga\downarrow,\, \lrbk{\Ga\vd t}_\ga\downarrow\\
    & \hspace*{1cm}\hfill \text{ and } \al=\lrbk{\Ga\vd t}_{\ga}  \in \lrbk{\Ga\vd A}_\ga }, \\[.3ex]
  \lrbk{\Ga,\indic{n}} & :=  \{\ga \,\lvert\, \ga\in\lrbk{\Ga}\}.
\end{align*}


\item The interpretation $\lrbk{\Ga\vd t}$ of a term $t$ in a context $\Ga$ is a partial function defined on $\lrbk{\Ga}$: Given $\ga \in \lrbk{\Ga}$,
\begin{align*}
  \lrbk{\Ga\vd \tprop}_\ga & :=  \{0,1 \}, \\  
  \lrbk{\Ga\vd \ttype_i}_\ga & :=  \calv_{\kappa_i}, \\
  \lrbk{\Ga\vd x}_{(\al_1,\dots,\al_n)} & :=  \al_i \quad\text{if $x$ is the $i$th declared variable in $\Ga$,}\tag{$\ast$}\\
  \lrbk{\Ga\vd \Pi x:A.B}_\ga & :=  \{\lam (f) : f \in \Pi_{\al \in\lrbk{\Ga\vd A}_\ga} \lrbk{\Ga, x:A\vd B}_{(\ga,\al)} \}, \\
  \lrbk{\Ga\vd \la x:A.t}_\ga & :=  \lam(\al\in\lrbk{\Ga\vd A}_\ga \mapsto \lrbk{\Ga,x:A\vd t}_{(\ga,\al)}), \\
  \lrbk{\Ga\vd t\,u}_\ga & :=  \app(\lrbk{\Ga\vd t}_\ga, \lrbk{\Ga\vd u}_\ga), \\
  \lrbk{\Ga\vd \lett{x:=t}{u}}_\ga & :=  \lrbk{\Ga, (x:=t:A)\vd u}_{\ga,\lrbk{\Ga\vd t}_\ga},\\
  & \text{where $A$ is such that $\lrbk{\Ga\vd t}\in\lrbk{\Ga\vd A}$,}
  \tag{$\dag$} \\
\lrbk{\Ga\vd \indic{n}\cdot z}_\ga & :=   \text{as explained above if defined,}\\
\lrbk{\Ga\edash\case{e,P,M_1,\dots,M_\ell}}_\ga & :=  \vec \app(\lrbk{\Ga\vd M_j},(\lrbk{\Ga\vd e})_{1},...,(\lrbk{\Ga\vd e})_{q} )\\
  & \text{if $(\lrbk{\Ga\vd e})_{0}=j$ where $\lh(e) = q+1$}, \\
\lrbk{\Ga\edash (\fix{f_j}{\oto{f/k:A:=t}} )}_\ga & :=  \text{as explained above if defined.}
\end{align*}

\noindent $(\ast)$ If $x \in \dom(\Ga)$, then the occurrence should be unique.

\noindent $(\dag)$ $A$ could be any term with the given property since the interpretation, when defined, is independent of it. 
  \end{enumerate}

\end{defi}

The following lemma is crucial for the soundness proof.

\begin{lem}[Substitutivity]\label{substitution}
Let $\Ga$ be a context and let $u, A$ be terms such that $\lrbk{\Ga\vd u}_\ga \in \lrbk{\Ga\vd A}_\ga$ for some $\ga\in \lrbk{\Ga}$ (assuming that both of them are defined), and write $\al=\lrbk{\Ga\vd u}_\ga$\,.
\begin{enumerate}[\em(1)]
\item Suppose $(\ga,\al),\de\in \lrbk{\Ga,x:A,\De}$. Then, $\ga,\de\in \lrbk{\Ga,\De[x\bs u]}$.
\item Suppose $(\ga,\al),\de\in \lrbk{\Ga,x:A,\De}$ and $\lrbk{\Ga,x:A,\De\vd t}_{(\ga,\al),\de} \downarrow$. Then,
  \begin{iteMize}{$\bullet$}
  \item $\lrbk{\Ga,\De[x\bs u]\vd t[x\bs u]}_{\ga,\de} \downarrow$.
  \item $\lrbk{\Ga,\De[x\bs u]\vd t[x\bs u]}_{\ga,\de}\, =\, \lrbk{\Ga,x:A,\De\vd t}_{(\ga,\al),\de}\, =\, \lrbk{\Ga,x:=u:A,\De\vd t}_{(\ga,\al),\de}$\,.
  \end{iteMize}
\end{enumerate}
\end{lem}
\begin{proof}
The assertions are proved for each $\De$ and $t$ by a mutual induction on the size of their arguments. In particular, given $\De$, the first assertion is proved before the second one for all $t$. In the case of $\De = [\,]$, the claims are obvious. Assume that $\De = \De_0, y:B$ and $\de = \de_0,\be$. The other cases can be considered similarly.
\begin{enumerate}[(1)]
\item $(\ga,\al),\de_0, \be \in \lrbk{\Ga,x:A,\De_0, y:B}$. Then, using the I.H. of the second claim, we have
  \[
    \be  \in  \lrbk{\Ga,x:A,\De_0\vd B}_{\ga,\al,\de_0}
     =    \lrbk{\Ga,\De_0[x\bs u]\vd B[x\bs u]}_{\ga,\de_0}\,.
  \]
That is, $\ga,\de_0,\be\in \lrbk{\Ga,\De_0[x\bs u],y:B[x\bs u]}$. 

\item We proceed by induction on $t$. If $t=x$, the claim follows because $\lrbk{\Ga\vd u}_\ga\downarrow$ implies that $\lrbk{\Ga,\De[x\bs u]\vd u}_{\ga,\de}\downarrow$ and $\lrbk{\Ga,\De[x\bs u]\vd u}_{\ga,\de} = \lrbk{\Ga\vd u}_\ga$. This is because the interpretation of $u$ does not depend on $\dom(\De)$. Other cases can be easily shown by using induction hypotheses.
\end{enumerate}
\end{proof}

\begin{thm}[Soundness]\label{thm-soundness} Our type system is sound with respect to the set-theoretic interpretation defined in Definition \ref{def-interpretation} in the following sense:
\begin{enumerate}[\em(1)]
\item If $\wf(\Ga)$, then $\lrbk{\Ga}$ is defined.
\item If $\Ga\ed M:A$, then $\lrbk{\Ga}$ is defined, and for any $\ga\in \lrbk{\Ga}$, it holds that $\lrbk{\Ga\vd M}_\ga$ and $\lrbk{\Ga\vd A}_\ga$ are defined, and that
\[
\lrbk{\Ga\vd M}_\ga \in \lrbk{\Ga\vd A}_\ga\,.
\]
\item If $\Ga\ed M=N:A$, then $\lrbk{\Ga}$ is defined, and for any $\ga\in \lrbk{\Ga}$, it holds that $\lrbk{\Ga\vd M}_\ga$, $\lrbk{\Ga\vd N}_\ga$, and $\lrbk{\Ga\vd A}_\ga$ are defined, and that
\[
\lrbk{\Ga\vd M}_\ga = \lrbk{\Ga\vd N}_\ga \in \lrbk{\Ga\vd A}_\ga\,.
\]
\item If $\Ga\edash M\prec N$, then $\lrbk{\Ga}$ is defined, and for any $\ga\in \lrbk{\Ga}$, it holds that $\lrbk{\Ga\vd M}_\ga$ and $\lrbk{\Ga\vd N}_\ga$ are defined, and that
\[
\lrbk{\Ga\vd M}_\ga \tm \lrbk{\Ga\vd N}_\ga \,.
\]
\end{enumerate}
\end{thm}
\begin{proof}
  We proceed by a simultaneous induction over the typing derivation. The cases $(wf), (ax), (var), (weak)$, and $(weak\tminus eq)$ are obvious.\medskip

\noindent $(\Pi)$ Suppose
\[
\infer{\Ga\edash \Pi x:A.B:s_3}{\Ga\edash A:s_1 & \Ga,x:A\edash B:s_2}\,.
\]
By I.H., it holds that $\lrbk{\Ga\vd A}_\ga\in \lrbk{\Ga\vd s_1}_\ga$ and $\lrbk{\Ga,x:A\vd B}_{\ga,\al}\in \lrbk{\Ga,x:A\vd s_2}_{\ga,\al}$ for all $\al\in \lrbk{\Ga\vd A}_\ga$. Now, we need to show that 
\[
\lrbk{\Ga\vd \Pi x:A.B}_\ga \in \lrbk{\Ga\vd s_3}_\ga\,.
\]
If $s_2 = s_3 = \tprop$, then Lemma \ref{A-poly} implies the claim. Assume $s_1 = \ttype_i, s_2 = \ttype_j, s_3 = \ttype_k$, and $i, j \le k$. Then, $\lrbk{\Ga \vd s_3} = \calv_{\ka_k}$, where $\ka_k$ is the $k$th inaccessible cardinal; hence, $\calv_{\ka_k}$ is closed under the power set operation.\medskip

\noindent The cases $(\Pi\tminus eq)$, $(\la)$, and $(\la\tminus eq)$ are obvious.\medskip
%

\noindent $(app)$ Suppose 
\[
\infer{\Ga\edash MN: B[x\bs N]}{\Ga\edash M:\Pi x:A.B & \Ga\edash N:A}\,.
\]
By induction hypothesis, it holds that $\lrbk{\Ga\vd N}_\ga \in \lrbk{\Ga\vd A}_\ga$ and $\lrbk{\Ga\vd M}_\ga=\lam(f)$ for some function $f$ with $\dom(f)=\lrbk{\Ga\vd A}$ and $f(\al)\in \lrbk{\Ga,x:A\vd B}_{\ga,\al}$ for any $\al\in \lrbk{\Ga\vd A}_\ga$. Thus, we have
\begin{align*}
  \lrbk{\Ga\vd MN}_\ga & = \app(\lrbk{\Ga\vd M}_\ga,\lrbk{\Ga\vd N}_\ga) 
 =  f(\lrbk{\Ga\vd N}_\ga)\\
 & \in \lrbk{\Ga,x:A\vd B}_{\ga,\lrbk{\Ga\vd N}_\ga}
 = \lrbk{\Ga\vd B[x\bs N]}_\ga\,.
\end{align*}
The cases $(app\tminus eq)$, $(let)$, and $(let\tminus eq)$ are similar.\medskip

The soundness of $(ind\tminus wf)$, $(ind\tminus type)$, and $(ind\tminus cons)$ are obvious from the interpretation constructions. The interpretations of inductive types and constructors are possible because of the induction hypotheses. This is the same for $(fix)$, $(fix\tminus eq)$, $(case)$, and $(case\tminus eq)$.\medskip

\noindent The cases $(ref), (sym)$, $(trans)$, $(conv)$, and $(conv\tminus eq)$ are obvious.\medskip

\noindent $(\be)$ Suppose
\[
\infer{\Ga\edash (\la x:A.M)N= M[x\bs N]: B[x\bs N]}{\Ga,x:A\edash M:B & \Ga\edash A:s_1 &  \Ga,x:A\edash B:s_2 & \Ga\edash N:A}\,.
\]
It remains to show that $\lrbk{\Ga\vd (\la x:A:M)N}_\ga = \lrbk{M[x\bs N]}_\ga$:
\begin{align*}
  \lrbk{\Ga\vd (\la x:A:M)N}_\ga & = \app(\lrbk{\Ga\vd \la x:A.M}_\ga, \lrbk{\Ga\vd N}_\ga)\\
 & = \app(\lam(\al\in \lrbk{\Ga\vd A}_\ga \mapsto \lrbk{\Ga,x:A\vd M}_{\ga,\al}),\lrbk{\Ga\vd N}_\ga)\\
 & = \lrbk{\Ga,x:A\vd M}_{\ga,\lrbk{\Ga\vd N}_\ga}\\
 & = \lrbk{M[x\bs N]}_\ga
\end{align*}
by Lemma \ref{substitution} because we know that $\lrbk{\Ga\vd N}_\ga \in \lrbk{\Ga\vd A}_\ga$ by induction hypothesis. The judgmental equality plays a crucial role in this case.\medskip

The case $(\delta)$ is obvious, and the case $(\zeta)$ follows from Lemma \ref{substitution} and induction hypothesis. The case $(\iota)$ is obvious by definition. The cumulativity rules are obviously sound. Finally, the soundness of the the constrained typing rules in Figure \ref{fig:constraint} follows directly from the arguments stated above. 
\end{proof}

\begin{thm}[Consistency]
  There is no term $t$ such that $\ed t:\Pi x:*.x$\,.
\end{thm}
\begin{proof}
 Note that  $\lrbk{\vd \Pi x:*.x}_{\tt nil} = \varnothing$\,.
\end{proof}

\section{Conclusion}\label{conclusion}
We identified some critical issues in constructing a set-theoretic, proof-irrelevant model of CC with cumulative type universes. Our construction reconfirmed that proof-irrelevance is a subtle and difficult subject to tackle when it is combined with the subtyping of the universes, in particular, $\tprop \prec \ttype$. We showed that the set-theoretic interpretation can be relatively easy when we work with judgmental equality. We believe that our study provides a (relatively) easy way for justifying the correctness of type theory in Martin-L\"of-styled, i.e., with simple model and, in particular, no proof of the strong normalization, which is usually very difficult to establish.

Besides the historical importance of Martin-L\"of-style type theory and the technical difficulties with external $\be$-reduction, there is another theoretical and practical reason for studying type systems with judgmental equality. In general, the equivalence of two systems with or without judgmental equality remains an open problem. Proving the equivalence of two systems with or without judgmental equality is not a simple task, even though some positive results have been achieved by \cite{coquand-algo}, \cite{goguen-phd, goguen-sound}, \cite{adams-eq}, and \cite{siles}. However, they are not sufficiently general to cover the case with cumulative type universes. Although \cite{adams-eq} mentioned that it might be possible to extend his proof to more general systems with unique principal types instead of type uniqueness as in the case of Luo's ECC \citep{luo-ecc, luo-utt} and Coquand's CIC, it still remains an open question.

A positive consequence of the work of \cite{adams-eq} and \cite{siles} is that the failed attempt of \cite{miwe}, i.e., without using sorted variables, would work if one first considers the system CC with judgmental equality and uses its equivalence to the usual CC. This is indeed the case for the model construction described in this paper, where we restrict the model construction to CC.

\section*{Acknowledgment}
We would like to express our sincere gratitude to Hugo Herbelin and Bruno Barras for their insightful discussions and advice. We would also like to thank the anonymous referees whose comments enabled us to improve this paper significantly.



\bibliographystyle{plain}
 



\appendix
\section{Definition of free variables and substitution}\label{A-defi}

The definitions of the sets of free variables in a context or term are standard. 
  \begin{align*}
    y[x\bs u] & := \begin{cases}
      u  & \text{if $x = y$},\\	
      y & \text{otherwise,}
    \end{cases}\\
    s[x\bs u] & := s,\\
    (\Pi y:A.B)[x\bs u] & := \begin{cases}
      \Pi y:A[x\bs u].B  & \text{if $x = y$},\\	
      \Pi y:A[x\bs u].(B[x\bs u]) & \text{otherwise,}\tag{$\ast$}
    \end{cases}\\ 
    (\la y:A.B)[x\bs u] & := \begin{cases}
      \la y:A[x\bs u].B  & \text{if $x = y$},\\	
      \la y:A[x\bs u].(B[x\bs u]) & \text{otherwise,}\tag{$\ast$}
    \end{cases}\\ 
    (\lett{y:=t_1}{t_2})[x\bs u] & :=
    \begin{cases}
      (\lett{y:=t_1[x\bs u]}{t_2}) & \text{if $x = y$},\\	
      \lett{y:=(t_1[x\bs u])}{t_2[x\bs u]} & \text{otherwise,}\tag{$\ast$}
    \end{cases}\\ 
    (t_1\, t_2)[x\bs u] & := (t_1[x\bs u]) (t_2[x\bs u]),\\
    \case{e,P,\vec f\, }[x\bs u] & := \case{e[x\bs u],P[x\bs u],\vec f[x\bs u]},\\
    (\indic{n}\cdot y)[x\bs u] & := \left
      \{\begin{array}{l}
        (\indic{n}\cdot y)\\
        \text{if $x \in \dom(\De_I,\De_C)$ or $FV(u)\cap\dom(\De_I,\De_C)\neq \varnothing$,}\tag{$\dag$}\\
        (\ind_n\!\single{\De_I[x\bs u]:=\De_C[x\bs u]}\cdot y) \text{ otherwise,}
      \end{array} \right .\\
    (\fix{y_i}{\oto{y / k: A:=t}})[x\bs u] & := \left \{
      \begin{array}{l}
        \fix{y_i}{\oto{y / k: A:=t}}\\
        \text{if $x\in\single{\vec y}$ or $FV(u)\cap\single{\vec y}\neq \varnothing$,}\tag{$\dag$}\\
        \fix{y_i}{\oto{y / k: (A[x\bs u]):=t[x\bs u]}} \text{ otherwise.}
      \end{array} \right .
  \end{align*}
$(\ast)$ By using $\al$-conversion, if needed, $y$ is assumed to be not free in $u$ such that the variable condition is satisfied.

\noindent $(\dag)$ The variable condition here implies that the names of inductive types, constructors, and recursive functions are uniquely determined, and that they will never be changed once they are defined. Thus, these names are bound variables that differ from variables bound by $\Pi$ and $\la$.

\section{Constrained typing}\label{constraint}
\begin{figure}[t]
\begin{center}
\noindent \begin{tabular}{|p{14.5cm}|}\hline \\[6.5ex]
\hfill \rbm{\infer{\Ga\edash\case{e,Q,(h_k)_{k}}:^{\consd} Q\, \vec{u}\, e}
  {\deduce{\Ga\edash h_k :^{\consd} \Pi(\vec v : \vec V_k)^{<\!\consc}.\, Q\, \vec{w_k}\,(c_k\,\vec p\, \vec v) \text{ for all } (c_k : \Pi\vec p : \vec P.\, \Pi\vec v : \vec V_k.\, d_i\, \vec p\, \vec w_k)\in \De_C}
    {\deduce{ \Ga\edash Q:^\epsilon B \quad \calc(d_i\,\vec{p}:A ; B) \quad \Ga\edash e :^{\consc} d_i \,\vec{p}\,\vec{u} }  
        {\indic{n}\in\Ga & (d_i:\Pi \vec p : \vec P.\, A)\in \De_I & \lh(\vec p)=n } } } }
\hfill $ $\\[4ex]

\hfill \rbm{ \infer{ \wf(\Ga, x:^\consc M)}
{\Ga\edash M: s } } 
 \hfill
\rbm{ \infer{ \Ga\edash x:^\consc A} 
  {\wf(\Ga) & x:^\consc A\in\Ga } }
\hfill \rbm{ \infer{ \Ga\edash t:^\epsilon A} 
  {\Ga\edash t:^\consc A } }
 \hfill $ $\\[4ex]

 \hfill \rbm{\infer{\Ga\edash \Pi x:^\consc A.B:^\epsilon s_3}
  {\Ga\edash A:^\epsilon s_1 & \Ga,x:^\consc A\edash B:^\epsilon s_2 & \calp(s_1,s_2,s_3)} }
\hfill $ $\\[4ex]
 
 \hfill \rbm{\infer{\Ga\edash \Pi x:^\consc A.B= \Pi x:^\consc A'.B':^\epsilon s_3}
   {\Ga\edash A= A':^\epsilon s_1 & \Ga,x:^\consc A\edash B= B':^\epsilon s_2 & \calp(s_1,s_2,s_3)} }
\hfill $ $\\[7ex]

\hfill \rbm{\infer{\Ga\edash \la x:^\consc A.M:^\consd \Pi x:^\consc A.B}
   {\deduce{\Ga, x:^\consc A\edash M:^\consd B }
     {\Ga\edash A :^\epsilon s_1 & \Ga,x:^\consc A\edash B:^\epsilon s_2} } } 
\hfill $ $\\[7ex]

\hfill \rbm{\infer{\Ga\edash \la x:^\consc A.M = \la x:^\consc A'.M':^\consd \Pi x:^\consc A.B}
  {\deduce{\Ga, x:^\consc A\edash M= M':^\consd B}
    {\Ga\edash A= A':^\epsilon s_1 & \Ga,x:^\consc A\edash B:^\epsilon s_2} } }
\hfill $ $\\[4ex]

 \hfill \rbm{\infer{\Ga\edash MN:^\consd  B[x\bs N]}
   {\Ga\edash M:^\consd \Pi x:^\consc A.B & \Ga\edash N:^\consc A}} 
\hfill $ $ \\[4ex] 

 \hfill \rbm{\infer{\Ga\edash MN= M'N':^\consd B[x\bs N]}
   {\Ga\edash M= M':^\consd\Pi x:^\consc A.B & \Ga\edash N= N:^\consc A} }
\hfill $ $\\[2ex]
\hline
\end{tabular}
\end{center}
\caption{Constrained typing}
\label{fig:constraint}
\end{figure}

Note that not all fix-point definitions can be accepted because of the possibility of non-normalizing terms. If one of the arguments belongs to an inductive type, then the function starts with a case analysis, and recursive calls are performed on variables coming from patterns and representing subterms. This is the usual restriction implemented in Coq when a case distinction with respect to a distinguished inductive type in a definition of a (mutual) recursive function occurs. These restrictions are imposed by the so-called \emph{guarded-by-destructors condition} defined by \cite{gimenez}. Here, we follow the simplified version given by \cite{pm-habili} by using \emph{constrained typing}. 

The constraints will be imposed with respect to a variable $z$ and an inductive specification $\De_I,\De_C$, and they have three forms: the \emph{empty constraint} $\epsilon$, the constraint $\sma z$, which describes the structural smallness with respect to $z$, and the constraint $\equ z$, which describes the equivalence to $z$. The constraints will be added to any occurrence of a variable in a term. Let $\consc,\consd,...$ vary over constraints. The judgments of constrained typing have the form $\Ga\edash M:^\consc N$, where the constraints are added to all the variables from $\dom(\Ga)$. $M^\epsilon$ and $\Ga^\epsilon$ denote the term $M$ and the term sequence $\Ga$, respectively, where only the constraint $\epsilon$ is added.

Given a constraint $\consc$, the constraint $<\!\consc$ is defined as follows:
\[
  <\epsilon := \epsilon\,,\qquad <\equ z := \sma z\,,\qquad <\sma z  := \sma z\,.
\]
The following defines the restriction of a recursive call of inductive type when defining a mutual recursion. Given  a declaration $\De$, $\De^{<\!z}$ is defined as follows:
\begin{align*}
([\,])^{<\! z} & := [\,], \\
(\De, x:A)^{<\! z} & := \De^{<\! z}, x:^{\epsilon} A \quad\,\,\, \text{if } FV(A) \cap \dom(\De_I, \De_C) = \varnothing, \\
(\De, x:A)^{<\! z} & := \De^{<\! z}, x:^{<\! z} A \quad \text{if } FV(A) \cap \dom(\De_I, \De_C) \neq \varnothing.
\end{align*}
$\Ga^{<\!z}$ is defined similarly for a term sequence $\Ga$. In Figure \ref{fig:constraint}, we list the rules for constrained typing. The omitted rules contain only the empty constraint $\epsilon$.

\section{Interpretation of Inductive types}\label{A-inductive}
Suppose $\Ga \vd \cald$, where $\cald = \indic{n}$, and $\ga \in \lrbk{\Ga}$. As mentioned before, we suppress $\Ga$ and $\ga$ for better readability. Suppose that 
\[
\begin{array}{c}
  \De_I := d_0:A_0,..., d_\ell:A_\ell\,,\,\,
  \De_C := c_1:T_1,..., c_m:T_m\,, \\[2ex]
  A_i  :=  \Pi\vec p : \vec P.\, \Pi\vec b_i : \vec B_i .\, s_i\,,\,\,
  T_k  := \Pi\vec p : \vec P.\, \Pi\vec z_k : \vec Z_k.\, d_{i_{k}}\, \vec p\, \vec t_k \,, \,\,
  Z_{k,j} := \Pi \vec u_{k,j} : \vec H_{k,j} .\, d_{i_{k,j}}\, \vec p\,\, \vec w_{k,j}\,,
\end{array}
\]
where $\lh(\vec p)=n$,
$\lh(\vec B_i) = \ell_i$,
$\lh(\vec t_k) = \ell_{i_k}$,  
$j\in\nu(k):=\menge{j}{FV(Z_{k,j})\cap\dom(\De_I)\neq\varnothing}$, and 
$i_k, i_{k,j}\le \ell$\,. 
Furthermore, $\vec{Z'_k}$ is defined as
\[
Z'_{k,j}:= \left \{
         \begin{array}{ll} 
           Z_{k,j} & \text{ if }j\not\in\nu_k,\\[1ex]
           \Pi \vec u_{k,j} : \vec H_{k,j} .\, d'_{i_{k,j}} & \text{ if }j\in\nu_k,
         \end{array}\right .
\]
where $d'_{i_{k,j}}$ are fresh variables. Further, we suppose that $\lrbk{\vec P}, \lrbk{\vec{Z'_k}}_{\rho_k}$, ... are already well defined below in the definition of $\Phi$ (This will be the case by induction hypothesis.).\smallskip

Then, we set $\lrbk{\cald}:=\cali(\Phi)$, where $\Phi :=$
\[
\bigcup_{i\le \ell}\,\, \bigcup_{k \in \mu_i} 
\Bigg \{
  \frac{\bigcup_{j\in\nu_k} \menge{\cons{i_{k,j} , \vec p, \lrbk{\vec w_{k,j}}_{\vec p, \vec z_k, \vec u}, \vec\app (z_{k,j}, \vec{u}\, )} }   
      {\vec{u}\in\lrbk{\vec H_{k,j}}_{\vec p, \vec z_k} }} 
    {\cons{i,\vec p, \lrbk{\vec t_{k}}_{\vec p, \vec z_k}, \cons{k, \vec z_k} }} 
  \,\,\,\,\Big \lvert\,\,  \vec p, \vec z_k \in \lrbk{\vec P}, \lrbk{\vec Z'_k}_{\rho_k} 
\Bigg \}\,.
\]
Here, $\mu_i:=\menge{k}{d_{i_k}= d_i}$, 
and $\rho_k$ associates $\calv_{\ka_{r(k,j)}}$ with $d'_{i_{k,j}}$, where $r(k,j):= rank(d_{i_{k,j}})$.\medskip

We also set $\lrbk{\cald\cdot d_i}:= \vec\lam(f_i)$, 
where
$f_i(\vec p, \vec b_i) = \calif(\Phi)(i,\vec p, \vec b_i)$
for $\vec p,\vec b_i:\lrbk{\vec P, \vec B_i}$,
and
$\lrbk{\cald\cdot c_k}:= \vec\lam(g_k)$, 
where
$g_k(\vec p, \vec z_k) = \cons{k,\vec z_k}$
for $\vec p, \vec z_k : \lrbk{\vec P, \vec Z_k}$\,.

\section{Interpretation of well-founded structured recursion}\label{A-recursion}

Below, we use the same notation as that used in Appendix \ref{A-inductive} for the inductive types on which the recursive call is running.\smallskip

Suppose $\Ga\edash \fix{f_\ell}{R} : A_j$, where
\[
\begin{array}{c}
R := \oto{f / k : A := t}, \quad A_i\equiv \Pi \vec x_i : \vec B_i.\, A'_i, \quad \lh(\vec B_i)=k_i+1, \quad \ell\le n, \\[2ex]
(\Ga\edash A_i: s_i)_{\forall i\le n}, \quad (\Ga, \vec f : \vec A \edash t_i : A_i)_{\forall i\le n}, \quad \calf(\vec f, \vec A, \vec k, \vec t) \,.
\end{array}
\]
Let $\ga \in \lrbk{\Ga}$ be given. We suppress $\Ga$ and $\ga$ for better readability. Then, $ \lrbk{\fix{f_\ell}{R}}$ will depend on the $\iota$-reduction. 

Suppose $\Ga \vd \vec a, a_{k_\ell+1} : \vec B_\ell$, where $a_{k_\ell+1} = T_k \, \vec p \, \vec u_k$, and $B_{\ell, k_{\ell}+1} = x_{i_\ell} \, \vec p \, \vec t_k$, i.e., $k\in \mu_{i_\ell}$, and that $\vec a, \vec p$ are all fresh variables, while $\vec u_k$ represents a branch in the tree-like structure. All the free variables occurring in $\vec u_k$ should be fresh. Then, $(t_\ell[f_i\bs (\fix{f_i}{R})])\, \vec a\, (T_k\, \vec p\, \vec u_k)$ $\beta$-, $\iota$-reduces to the term $M_{\ell,k}$ which is obtained from the node term of the branch which $T_k \, \vec p \, \vec u_k$ represents. 

Suppose that for some $g_{\ell,k}, h_{\ell, k}\in \setn$, $u'_{1}, ..., u'_{g_{\ell,k}}, u'_{g_{\ell,k}+1}, ..., u'_{h_{\ell, k}}$ list all the subterms of $M_{\ell,k}$ that are structurally smaller than $\vec u_k$. Each $u'_q$ with $q \le g_{\ell,k}$  occurs as the $(k_{n_q}+1)$th argument of $(\fix{f_{n_q}}{R})$. Thus,
\[
((\fix{f_{n_q}}{R})\,\vec b_{n_q}\, u'_{q})
\]
is a subterm of $M_{j,k}$ for some $\vec b_{n_q} : B_{n_q,1},...,B_{n_q,k_{n_q}}$. Note that each $u'_q$ is an argument of some constructor $T_{m_q} :  \Pi\vec p : \vec P.\, \Pi \vec z_{m_q} : \vec Z_{m_q}.\, x_{i_{m_q}}\, \vec p\, \vec{t}_{m_q}$ such that the head of $u'_q$ is of some type $Z_{m_q,j_q} = \Pi \vec u_{m_q,j_q} : \vec H_{m_q,j_q} .\, x_{i_{m_q,j_q}}\, \vec p\,\, \vec{w}_{m_q,j_q}$ and that $x_{i_{m_q,j_q}} = x_{i_{n_q}}$ if $q \le g_{\ell,k}$. Furthermore, we suppose that $u'_{q_1},..., u'_{q_h}$ are all terms among $u'_1, ..., u'_{h_{\ell,k}}$, which are headed by some variables.

Thus, $u'_{q_r} = a'_{q_r} \, \vec u'_{m_{q_r},j_{q_r}}$ for a variable $a'_{q_r}$ of type  
\[
Z_{m_{q_r},j_{q_r}}  = \Pi \vec u_{m_{q_r},j_{q_r}} : \vec H_{m_{q_r},j_{q_r}} .\, x_{i_{m_{q_r},j_{q_r}}}\, \vec p\,\, \vec{w}_{m_{q_r},j_{q_r}}
\]
and for some terms $\vec u'_{m_{q_r},j_{q_r}}$. Note that $a'_{q_1}, ..., a'_{q_h}$ are exactly the free variables occurring in $\vec u_k$. Suppose that $u'_{q_{r_1}}, ..., u'_{q_{r_q}}$ are all such terms structurally smaller than $u'_q$. Then,
\[
u'_q = 
(\la \vec u_{m_{q_{r_1}},j_{q_{r_1}}} : \vec H_{m_{q_{r_1}},j_{q_{r_1}}}.\, ...\, 
\la \vec u_{m_{q_{r_q}},j_{q_{r_q}}} : \vec H_{m_{q_{r_q}},j_{q_{r_q}}} .\, \tilde u'_q)\, 
\vec u'_{m_{q_{r_1}},j_{q_{r_1}}} \cdots \vec u'_{m_{q_{r_q}},j_{q_{r_q}}}
\]
for some $\tilde u'_q$. Similarly, $M_{\ell,k}$ can be written as follows:
\[
(\la \vec u_{m_{q_{r_1}},j_{q_{r_1}}} : \vec H_{m_{q_{r_1}},j_{q_{r_1}}}.\, ...\, 
\la \vec u_{m_{q_{r_q}},j_{q_{r_q}}} : \vec H_{m_{q_{r_q}},j_{q_{r_q}}} .\, 
(\fix{f_{n_q}}{R})\,\vec b_{n_q}\, \bar u'_{q}  )\, 
\vec u'_{m_{q_{r_1}},j_{q_{r_1}}} \cdots \vec u'_{m_{q_{r_q}},j_{q_{r_q}}}
\]
where 
\[
\bar u'_q = \la \vec u_{m_{q_{r_1}},j_{q_{r_1}}} : \vec H_{m_{q_{r_1}},j_{q_{r_1}}}.\, ...\, 
\la \vec u_{m_{q_{r_q}},j_{q_{r_q}}} : \vec H_{m_{q_{r_q}},j_{q_{r_q}}} .\, 
\tilde u'_q\,.
\]
Further, set
\[
\bar M_{\ell,k} = 
(\la \vec u_{m_{q_{r_c}},j_{q_{r_c}}} : \vec H_{m_{q_{r_c}},j_{q_{r_c}}})_c .\,
(\fix{f_{n_q}}{R})\,\vec b_{n_q}\, \bar u'_{q}  
\]
of type 
$(\Pi \vec u_{m_{q_{r_c}},j_{q_{r_c}}} : \vec H_{m_{q_{r_c}},j_{q_{r_c}}})_c .\,
A'_{n_q}\single{\vec x_{n_q} : \vec b_{n_q}, u'_q}$,
where $c$ ranges over $1,...,q$.
Lastly, let $M'_{\ell,k}$ be obtained from $M_{\ell,k}$ by replacing $\bar M_{\ell,k}$ with a fresh variable $X_{n_q}$.\smallskip

Then, $\displaystyle \bigcup_{\ell} \,\lrbk{\fix{f_\ell}{R}}$ will  correspond to the fixpoint of the following rule set:
\begin{eqnarray*}
\Psi & := & \bigcup_{\ell\le n}\,\, \bigcup_{k \in \mu_{i_\ell}} 
\Bigg \{
  \frac{\bigcup_{q\in \single{1,...,g_{\ell,k}}} 
             \menge{\cons{ \lrbk{\vec b_{n_q}}_{\rho}, \lrbk{u'_{q}}_{\rho} ,  \vec\app (v'_{q}, \vec u_{m_{q_{r_1}}, j_{q_{r_1}}},...,\vec u_{m_{q_{r_q}}, j_{q_{r_q}}} )} }  {\calc} }
         {\cons{\al_1,...,\al_{k_\ell}, \cons{k, \lrbk{\vec{u}_k}_{\rho}}, \lrbk{M'_{\ell,k}}_{\eta} } } \,\,\,\,\Big \lvert \\
  & & \hspace{1.9cm} \calc \equiv \vec u_{m_{q_{r_1}}, j_{q_{r_1}}} \in \lrbk{\vec H_{m_{q_{r_1}}, j_{q_{r_1}}}},..., \vec u_{m_{q_{r_q}}, j_{q_{r_q}}} \in \lrbk{\vec H_{m_{q_{r_q}}, j_{q_{r_q}}}},\\[1ex]
  & & \hspace{1.9cm} \vec \al \in \lrbk{B_{\ell,1}},...,\lrbk{B_{\ell,k_\ell}},\\
  & & \hspace{1.9cm} v_{q_r} \in \lrbk{\Pi \vec{u}_{m_{q_r},j_{q_r}}: \vec H_{m_{q_r},j_{q_r}} .\, x_{i_{m_{q_r},j_{q_r}}}\, \vec p\,\, \vec{w}_{m_{q_r},j_{q_r}} },\\
  & & \hspace{1.9cm} \text{$\rho$ associates $\vec\al$ to $\vec a$, $v_{q_r}$ to $a'_{q_r}$, $r \in \single{1,...,h}$,}\\
  & & \hspace{1.9cm} v'_{q} \in \lrbk{(\Pi \vec u_{m_{q_{r_c}},j_{q_{r_c}}}  : \vec H_{m_{q_{r_c}},j_{q_{r_c}}})_c .\, 
  A'_{n_q}\single{\vec x_{n_q} : \vec b_{n_q}, u'_q} }_{\rho},\\
  & & \hspace{1.9cm} \text{$\eta$ associates $\vec\al$ to $\vec a$, $v_{q_r}$ to $a'_{q_r}$, $r \in \single{1,...,h}$, and $v'_{q}$ to $X_{n_q}$}.
\Bigg \} 
\end{eqnarray*}
We set $\lrbk{\fix{f_\ell}{R}} := \vec \lam (h)$, where $h$ is a function such that
\[
h(a_1,...,a_{k_\ell},\cons{k,\vec z_k}) = \calif (\Psi) (a_1,...,a_{k_\ell},\cons{k,\vec z_k})
\]
where $\vec a, \cons{k,\vec z_k} \in \lrbk{\vec B_\ell}$.

%

\end{document}